\documentclass[letterpaper,11pt]{article}
\usepackage[utf8]{inputenc}
\usepackage[T1]{fontenc}

\pdfoutput=1

\usepackage{amsmath, amsthm, amssymb, thm-restate}
\usepackage{algorithmicx}
\usepackage[table,xcdraw]{xcolor}
\usepackage[ruled,vlined,linesnumbered]{algorithm2e}
\usepackage{setspace}
\usepackage{mathtools}
\usepackage[style=alphabetic,natbib=true,maxnames=99,maxalphanames=5]{biblatex}
\usepackage{comment} 
\usepackage{tcolorbox} 
\usepackage{xfrac}
\usepackage{hyperref}
\usepackage{multirow}
\usepackage{caption}
\usepackage{bm}
\usepackage{newfloat}
\usepackage{enumitem}
\usepackage{dblfloatfix} 
\usepackage{wrapfig}
\usepackage{csquotes}
\usepackage{mathrsfs}
\usepackage{float}
\usepackage{tikz}
\usetikzlibrary{positioning, calc}

\usepackage{fancyhdr}

\renewcommand{\epsilon}{\varepsilon}

\DeclareMathOperator{\E}{\ensuremath{\normalfont \textbf{E}}}

\newcommand{\eps}{\epsilon}
\newcommand{\poly}{\textnormal{poly}}

\renewcommand{\P}{\mathcal{P}}
\newcommand{\ph}{\hat{p}}
\newcommand{\Eh}{\hat{E}}

\newcommand{\Tt}{\tilde{T}}
\newcommand{\Th}{\hat{T}}
\newcommand{\Et}{\tilde{E}}
\renewcommand{\Eh}{\hat{E}}
\newcommand{\Ft}{\tilde{F}}
\newcommand{\wt}{\tilde{w}}

\newcommand{\MST}{\textnormal{MST}}
\newcommand{\TSP}{\textnormal{TSP}}
\newcommand{\C}{\mathcal{C}}
\newcommand{\D}{\mathcal{D}}
\newcommand{\A}{\mathcal{A}}
\renewcommand{\emptyset}{\varnothing}

\newcommand{\Boruvka}{Bor\r{u}vka}
\renewcommand{\H}{\mathcal{H}}
\renewcommand{\P}{\mathcal{P}}
\newcommand{\Ph}{\widehat{\P}}
\let\LL\L
\renewcommand{\L}{\mathcal{L}}
\newcommand{\sP}{\P}
\newcommand{\sPh}{\hat{\sP}}
\newcommand{\sC}{\C}
\newcommand{\mc}[1]{\mathcal{#1}}
\newcommand{\MPX}{{\textnormal{\tiny MPX}}}
\newcommand{\onetwocycle}[0]{{\normalfont 1vs2\textsc{-Cycle}}}

\newcommand{\ceil}[1]{{\left\lceil{#1}\right\rceil}}

\newcommand{\card}[1]{\lvert#1\rvert}

\DeclareMathOperator*{\Prob}{\ensuremath{\textnormal{Pr}}}
\renewcommand{\Pr}{\Prob}

\DeclareMathOperator*{\argmin}{arg\,min}

\newcommand{\Ot}{\ensuremath{\widetilde{O}}}

\usepackage[noabbrev,nameinlink]{cleveref}
\crefname{lemma}{Lemma}{Lemmas}
\crefname{theorem}{Theorem}{Theorems}
\crefname{property}{Property}{Properties}
\crefname{claim}{Claim}{Claims}
\crefname{result}{Result}{Results}
\crefname{conj}{Conjecture}{Conjectures}
\crefname{definition}{Definition}{Definitions}
\crefname{observation}{Observation}{Observations}
\crefname{proposition}{Proposition}{Propositions}
\crefname{assumption}{Assumption}{Assumptions}
\crefname{line}{Line}{Lines}
\crefname{figure}{Figure}{Figures}
\creflabelformat{property}{(#1)#2#3}
\crefname{equation}{}{}
\crefname{section}{Section}{Sections}
\crefname{appendix}{Appendix}{Appendices}
\crefname{problem}{Problem}{Problems}
\crefname{algCounter}{Algorithm}{Algorithms}
\Crefname{algCounter}{Algorithm}{Algorithms}

\newtheorem{problem}{Problem}
\newtheorem{theorem}{Theorem}

\newtheorem{lemma}{Lemma}[section]
\newtheorem{proposition}[lemma]{Proposition}

\newtheorem{conj}{Conjecture}
\newtheorem{definition}[lemma]{Definition}
\newtheorem{claim}[lemma]{Claim}

\newtheorem{remark}{Remark}
\newtheorem*{remark*}{Remark}

\definecolor{mylightgray}{RGB}{230,230,230}

\algnewcommand{\IIf}[2]{\textbf{if} #1 \textbf{then} #2}
\algnewcommand{\EndIIf}{\unskip\ \algorithmicend\ \algorithmicif}

\newenvironment{graytbox}{
\par\addvspace{0.1cm}
\begin{tcolorbox}[width=\textwidth,
                  boxsep=5pt,
                  left=1pt,
                  right=1pt,
                  top=2pt,
                  bottom=2pt,
                  boxrule=0pt,
                  arc=0pt,
                  colback=mylightgray,
                  colframe=black,
                  ]
}{
\end{tcolorbox}
}

\newenvironment{whitetbox}{
\par\addvspace{0.1cm}
\begin{tcolorbox}[width=\textwidth,
                  boxsep=5pt,
                  left=1pt,
                  right=1pt,
                  top=2pt,
                  bottom=2pt,
                  boxrule=1pt,
                  arc=0pt,
                  colframe=black,
                  colback=white
                  ]
}{
\end{tcolorbox}
}

\newcounter{algCounter}

\newcounter{myalgcounter}
 
\crefname{myalgcounter}{Algorithm}{Algorithms}
\Crefname{myalgcounter}{Algorithm}{Algorithms}

\usepackage[margin=1in]{geometry}

\allowdisplaybreaks

\definecolor{mygreen}{RGB}{10,150,110}
\definecolor{myred}{RGB}{150,10,20}

\hypersetup{
     colorlinks=true,
     citecolor= mygreen,
     linkcolor= myred
}

\makeatletter
\renewcommand{\paragraph}{%
  \@startsection{paragraph}{4}%
  {\z@}{10pt}{-1em}%
  {\normalfont\normalsize\bfseries}%
}
\makeatother

\makeatletter
\patchcmd{\@algocf@start}
  {-1.5em}
  {0pt}
  {}{}
\makeatother

\title{Massively Parallel Minimum Spanning Tree \\ in General Metric Spaces}

\author{
Amir Azarmehr\\{\em Northeastern University} \and 
Soheil Behnezhad \\{\em Northeastern University} \and
Rajesh Jayaram \\{\em Google Research}
\and
Jakub \LL{}\k{a}cki \\{\em Google Research}
\and
Vahab Mirrokni \\{\em Google Research}
\and
Peilin Zhong \\{\em Google Research}
}

\date{}

\begin{document}

\maketitle

\thispagestyle{empty}
\begin{abstract}
\parskip=5pt

    We study the minimum spanning tree (MST) problem in the massively parallel computation (MPC) model. Our focus is particularly on the {\em strictly sublinear} regime of MPC where the space per machine is $O(n^\delta)$. Here $n$ is the number of vertices and constant $\delta \in (0, 1)$ can be made arbitrarily small. The MST problem admits a simple and folklore $O(\log n)$-round algorithm in the MPC model. When the weights can be arbitrary, this matches a conditional lower bound of $\Omega(\log n)$ which follows from a well-known \onetwocycle{} conjecture. As such, much of the literature focuses on breaking the logarithmic barrier in more structured variants of the problem, such as when the vertices correspond to points in low- \cite[STOC'14]{AndoniNikolov} or high-dimensional Euclidean spaces \cite[SODA'24]{jayaram2024massively}.
    
    In this work, we focus more generally on metric spaces. Namely, all pairwise weights are provided and guaranteed to satisfy the triangle inequality, but are otherwise unconstrained. We show that for any $\varepsilon > 0$, a $(1+\varepsilon)$-approximate MST can be found in $O(\log \frac{1}{\varepsilon} + \log \log n)$ rounds, which is the first $o(\log n)$-round algorithm for finding any constant approximation in this setting. Other than being applicable to more general weight functions, our algorithm also slightly improves the $O(\log \log n \cdot \log \log \log n)$ round-complexity of \cite[SODA'24]{jayaram2024massively} and significantly improves its approximation from a large constant to $1+\varepsilon$.
    
    On the lower bound side, we prove that under the \onetwocycle{} conjecture, $\Omega(\log \frac{1}{\varepsilon})$ rounds are needed for finding a $(1+\varepsilon)$-approximate MST in general metrics. This implies that $(i)$ the \mbox{$\varepsilon$-dependency} of our algorithm is optimal, $(ii)$ it is necessary to {\em approximate} MST in order to beat $\Omega(\log n)$ rounds in the metric case, and $(iii)$ computing metric MST is strictly harder than computing MST in low-dimensional Euclidean spaces.

    It is also worth noting that while many existing lower bounds in the MPC model under the \onetwocycle{} conjecture only hold against ``component-stable'' algorithms, our lower bound applies to {\em all} algorithms. Indeed, a conceptual contribution of our paper is to provide a natural way of lifting the component stability assumption which we hope to have other applications.
\end{abstract}

{
\clearpage
\tableofcontents{}
\vspace{1cm}
\renewcommand{\baselinestretch}{1}
\setcounter{tocdepth}{2}
\thispagestyle{empty}
\clearpage
}

\setcounter{page}{1}

\clearpage

\section{Introduction}

The minimum spanning tree (MST) problem is one of the most 
fundamental problems in combinatorial optimization, the 
algorithmic study of which dates back to the work of \Boruvka{} in 1926~\cite{boruuvka1926jistem}. Given a set of points and distances between the points, the goal is to compute a  minimum-weight tree over the points. The MST problem has received a tremendous amount of attention from the algorithm design community, and has been studied in a wide variety of computational models, 
leading to a large toolbox of methods for the problem~\cite{Charikar:2002,IT03,indyk2004algorithms,10.1145/1064092.1064116,har2012approximate,AndoniNikolov,andoni2016sketching,bateni2017affinity,czumaj2009estimating,CzumajEFMNRS05,CzumajS04,chazelle00,chazellerubinfeld}.


To deal with the sheer size of these modern embedding datasets, the typical approach is to implement algorithms in massively parallel computation systems such as MapReduce~\cite{dean2004mapreduce,dean2008mapreduce}, Spark~\cite{zaharia2010spark}, Hadoop~\cite{white2012hadoop}, Dryad~\cite{isard2007dryad} and others. The \emph{Massively Parallel Computation (MPC)} model~\cite{karloff2010model,goodrich2011sorting,beame2017communication,AndoniNikolov} is a computational model for these systems that balances accurate modeling with theoretical elegance.

\paragraph{The MPC Model:} The input, which in the case of the MST problem, is the edge set of a weighted graph $G=(V, E)$ with $n$ vertices and $m$ edges, is initially distributed across $M$ machines. Each machine has space $S = n^\delta$ words, where constant $\delta \in (0, 1)$ can be made arbitrarily small.\footnote{This is known as the {\em strictly sublinear} variant of the MPC model. The easier regime of MPC where the space per machine is $S=O(n)$ or slightly super linear in $n$ has also been studied extensively in the literature.} Ideally, we set $M = \widetilde{O}(m/S)$ so that the total space across all machines (i.e., $M \cdot S$) is linear in the input size $O(m)$, ensuring that there is just enough space to store the entire input. Computation is performed in synchronous {\em rounds}. During each round, each machine can perform any computation on its local data and send messages to other machines. Messages are delivered at the beginning of the next round. A crucial constraint is that the total size of messages sent and received by each machine in a round must be at most $S$. The primary goal is to minimize the number of rounds, which is the main bottleneck in practice.

The MST problem, in particular, has been extensively studied in the MPC model \cite{AndoniNikolov,bateni2017affinity,karloff2010model,andoni2018parallel,lattanzi2011filtering,Nowicki-STOC21,ahanchi2023massively,chen2022new,jayaram2024massively}. For arbitrary edge-weights, algorithms that obtain an exact MST in $O(\log n)$ rounds are known using connected components algorithms \cite{karloff2010model,andoni2018parallel,behnezhad2019near,coy2022deterministic}. While an $O(\log n)$-round algorithm is often considered too slow in the MPC model (see e.g. \cite{andoni2018parallel,behnezhad2019near} and the references therein), improving the round complexity to sublogarithmic is unlikely. In particular, it is not hard to see that such an algorithm would refute the following widely believed \onetwocycle{} conjecture {\cite{yaroslavtsev2018massively,roughgarden2018shuffles,andoni2018parallel,behnezhad2019near,assadi2019massively}}:

\begin{conj}[The \onetwocycle{} Conjecture]
    There is no MPC algorithm with $S=n^{1-\Omega(1)}$ local space and $M=\poly(n)$ machines that can distinguish whether the input is a cycle on $2n$ vertices or two cycles on $n$ vertices each in $o(\log n)$ rounds.
\end{conj}

Because of this conditional lower bound, much of the work in the literature of massively parallel MST algorithms has been on more structured classes of graphs, particularly metric spaces.

\paragraph{Massively Parallel MST in Metric Spaces:} A decade ago, \citet{AndoniNikolov} showed that if the vertices correspond to points in a $d$-dimensional Euclidean space for $d = o(\log n)$, then a $(1+o(1))$-approximate MST can be found in $O(1)$ rounds. For high-dimensional Euclidean spaces, where $d = \Omega(\log n)$, \citet{jayaram2024massively} showed recently that an $O(1)$-approximate MST can be found in $O\big(\log \log n \cdot \log \log \log n \big)$ rounds, breaking the logarithmic-round barrier.

Unfortunately, the progress on the Euclidean version of the problem does not aid in the construction of algorithms for more general metric spaces where the input specifies the pairwise weights $w: \binom{n}{2} \to \mathbb{R}_{+}$ and these weights are only guaranteed to satisfy the triangle inequality. In particular, both the algorithms of \cite{jayaram2024massively}  and \cite{AndoniNikolov} crucially rely on the construction of space partitions of Euclidean space, via $\eps$-nets for low-dimensional space, and hypergrid-decompositions for high-dimensional space. These partitions can be efficiently constructed in the MPC model, and are crucial to avoiding $O(\log n)$ round algorithms, as they allow one to bypass running complete connected component algorithms by exploiting the geometry of the space. However, in the general metric case, it is not clear how to exploit the geometry, as we have neither bounded size $\eps$-nets or natural counterparts to hypergrid partitions. Thus it is not clear whether the hardness of the MST problem extends to metric inputs, or whether, as for the Euclidean case, the metric MST problem admits sublogarithmic round algorithms.  
Specifically,  we address the following question:

\begin{quote}
	\begin{center}
		{\it  Do there exist $o(\log n)$ round MPC algorithms for MST in general metrics?}
	\end{center}
\end{quote}

\subsection{Our Contribution}

In this work, we address the above question affirmatively, by proving the following theorem:

\begin{graytbox}
\begin{theorem} \label{thm:main}
	Given a metric, for any fixed $\delta \in (0, 1)$ and any $\epsilon > 0$, a $(1 + \epsilon)$-approximate MST can be computed in $O\left(\log \frac{1}{\epsilon} + \log \log n\right)$ rounds of the MPC model, with $O(n^\delta)$ space per machine and $\widetilde{\Theta}(n^2)$ total space which is near-linear in the input size.
\end{theorem}
\end{graytbox}

Other than being applicable more generally to metric spaces, the algorithm of \cref{thm:main} also improves over the algorithm of \cite{jayaram2024massively} by shaving off a $(\log \log \log n)$ term from its round-complexity and improving its approximation from a large constant to $(1+\epsilon)$.
Also, we note that \cref{thm:main}, when combined with existing techniques, readily leads to an algorithm for the $(2 + \epsilon)$-approximation of TSP with the same round and space complexity (see \cref{sec:tsp}). This is possible due to the hierarchical structure of the MST computed by the algorithm.

\begin{remark}
    When $d \ll n$, our \cref{thm:main} uses more total space than the $d$-dimensional Euclidean space algorithms of \cite{jayaram2024massively,AndoniNikolov}. However, since the input has size $\Theta(n^2)$ for general metrics (as opposed to $O(nd)$ for $d$-dimensional Euclidean spaces) this is unavoidable. One may wonder if a more succinct representation of distances, say by taking shortest path distances over a sparse base graph, can reduce total space while keeping the round-complexity sublogarithmic. This is, unfortunately, not possible since the \onetwocycle{} instance as the base graph implies that $\Omega(\log n)$ rounds are necessary to tell if the MST cost is $2n-1$ or $\infty$.
\end{remark}

Our next contribution is to prove the following lower bound for computing approximate MSTs in general metrics.
\begin{graytbox}
\begin{restatable}{theorem}{thmlb}\label{thm:dynamic-final}
    \label{thm:lb}
	Under the \onetwocycle{}  conjecture, 
	any MPC algorithm with $n^{1-\Omega(1)}$  space per machine and $\poly(n)$ total space requires $\Omega\left(\log \frac{1}{\epsilon}\right)$ rounds to compute a $(1 + \epsilon)$-approximate MST for any $1/n \leq \epsilon \leq 1$, even in $(1, 2)$-metrics where all distances are either 1 or 2.
\end{restatable}
\end{graytbox}

Assuming the \onetwocycle{} conjecture, \cref{thm:lb} has the following implications: 

\begin{itemize}
    \item It shows that the $\epsilon$-dependency of our \cref{thm:main} is optimal. In particular, for $\epsilon \leq 1/\log n$, the upper and lower bounds of \cref{thm:main,thm:lb} match.
    \item It shows that a $(1+o(1))$-approximate MST of a general metric cannot be computed in $O(1)$ rounds (under \onetwocycle{}). Note that this is unlike the algorithm of \cite{AndoniNikolov} which obtains $(1+o(1))$-approximation of a $d$-dimensional Euclidean space in $O(1)$ rounds whenever $d = o(\log n)$. Therefore, computing MST in general metrics is strictly harder than in low-dimensional Euclidean spaces.
    \item It shows that {\em approximation} is necessary in order to achieve a sub-logarithmic round algorithm for MST in general metrics.
\end{itemize}

Finally, we believe that our techniques in proving \cref{thm:lb} might be of interest for other MPC lower bounds and particularly getting rid of a common ``component stability'' assumption. We discuss this in more detail in the following \cref{sec:component-stable-intro}.

\subsection{MPC Lower Bounds: Beyond Component Stability}\label{sec:component-stable-intro}

A  powerful approach for proving conditional MPC lower bounds for various graph problems is to combine {\em distributed LOCAL} lower bounds with the \onetwocycle{} conjecture \cite{GhaffariKU19,CzumajDP21a}. On the one hand, an $\Omega(t)$ round distributed LOCAL lower bound for a graph problem implies that vertices need to see $\Omega(t)$-hops away to decide on their output for that problem. On the other hand, the \onetwocycle{} conjecture essentially implies that to discover vertices that are $\Omega(t)$-hops away, one needs $\Omega(\log t)$ rounds in the sublinear MPC model. Thus, intuitively, an $\Omega(t)$ distributed LOCAL lower bound for a problem $\mc{P}$ should imply an $\Omega(\log t)$ round lower bound for $\mc{P}$ in MPC under the \onetwocycle{} conjecture.  

To make the connection above formal, existing methods \cite{GhaffariKU19,CzumajDP21a} require an additional assumption that the output of the MPC algorithm on one connected component should be independent of its output on other connected components---a property that is referred to as {\em component stability} in the literature. While many MPC algorithms are indeed component-stable, there are natural ones that are not---we refer interested readers to the paper of \citet{CzumajDP21a} for examples of such algorithms and other limitations of component stability. So, ideally, we would like to get rid of the component stability assumption. Unfortunately, there is still a huge gap between conditional lower bounds against component-stable and non-stable algorithms. We believe that our techniques in proving \cref{thm:lb} are a step forward towards proving MPC lower bounds without making the component stability assumption.

We show in \cref{subsec:lb-local-random} that any distributed local algorithm needs at least $\Omega(1/\epsilon)$ rounds for computing a $(1+\epsilon)$-approximate maximum matching in  $(1, 2)$-metrics (when the communication network is the one induced on the weight-1 edges). Using the framework of \cite{CzumajDP21a}, one can lift this to an MPC lower bound, showing that any \underline{component-stable} MPC algorithm needs $\Omega(\log 1/\epsilon)$ rounds to find a $(1+\epsilon)$-approximate MST in $(1, 2)$-metrics (see \cref{sec:lift-to-MPC}). But note that our \cref{thm:lb} applies to all MPC algorithms, not just component-stable ones. Instead of component stability, our proof of \cref{thm:lb} heavily relies on the fact that if two vertices $u$ and $v$ are {\em isomorphic} to one another, an MPC algorithm should have the same output distribution for both. Note that unlike component stability, this assumption comes without loss of generality as one can randomly relabel the vertex IDs and the algorithm cannot distinguish $u$ from $v$. We hope that our techniques for proving \cref{thm:lb} can pave the way for proving conditional (under \onetwocycle{} conjecture) MPC lower bounds against all MPC algorithms, not just component-stable ones.

\section{Technical Overview}

We present informal overviews of our algorithms and lower bounds of \cref{thm:main,thm:lb} in \cref{sec:tech-ub,sec:tech-lb}. Since the techniques in our upper and lower bounds are completely different, these two sections are independent and can be read in any order.

\subsection{Technical Overview of \cref{thm:main}}\label{sec:tech-ub}

We start with a by-now-standard framework for computing MST based on threshold graphs \cite{chazellerubinfeld,CzumajS04,Nowicki-STOC21,jayaram2024massively} and then discuss the challenges that arise in using it in the MPC model.

Given a metric space over $n$ vertices $V$ where the distance between $u, v\in V$ is denoted as $w(u,v)\in[1, W]$ ($W=\poly(n)$), the MST of $V$ can be obtained by the following simple variant of the Kruskal's algorithm. Let $G_{\leq t}$ be a threshold graph over $V$ where there is an edge between vertices $u, v \in V$ iff $w(u, v)\leq t$. 
Let $\P_t$ be the partition of $V$ where each vertex set in $\P_t$ is a connected component of $G_{\leq t}$.
Then $\P_1, \P_2,\P_4,\ldots,\P_{W}$ form a hierarchy of partitions for $V$, i.e., if $u$ and $v$ are in the same set of $\P_t$, they are also in the same set of $\P_{t'}$ for all $t'\geq t$, and $\P_{W}$ contains only one vertex set containing all the vertices.

For each $t\geq 1$, consider a set $C$ of $\P_t$.
Suppose $C$ is the union of sets $C_1,C_2,\ldots, C_r$ from $\P_{t/2}$ (define $\P_{1/2}$ as the trivial partition where each set is a singleton).
We regard each $C_i$ as a \emph{super-node} and compute an MST over $C_1,C_2,\ldots, C_r$.
Then, the union of the edges found by the above process for all $t$ and $C$ forms an MST for the whole metric.
Note that every edge between $C_i$ and $C_j$ in $G_{\leq t}$ must be in the range $[t/2, t]$.
This implies that if we compute any arbitrary spanning tree over $C_1,C_2,\ldots, C_r$ in $G_{\leq t}$ (instead of an MST), and take the union of all edges found, it gives a $2$-approximate MST. The approximation can be further improved to $(1+\epsilon)$ by setting the ratio between the thresholds to $(1+\epsilon)$ instead of $2$.

The problem with the framework above is that under the \onetwocycle{} conjecture, it is actually impossible to construct the connected components of each threshold graph in $o(\log n)$ rounds exactly (the graph $G_{\leq 1}$ may itself be the \onetwocycle{} instance). We first overview the approach of \cite{jayaram2024massively} which essentially approximates these components in the Euclidean space via \emph{shifted grids}. We then show how we go beyond shifted grids (which only exist for Euclidean spaces and not general metrics) and solve the problem in general metrics.

\vspace{-0.2cm}
\subsubsection*{MST in Euclidean Spaces via Shifted Grids: An Overview of \cite{jayaram2024massively}}

Motivated by the above MST approximation algorithm, \cite{jayaram2024massively} made the following observation: 
Suppose $\Ph_1,\Ph_2,\Ph_4,\ldots,\Ph_{W}$ is a randomized hierarchical partitioning of $V$ such that $(1)$ each vertex set in $\Ph_t$ is a subset of some connected component of $G_{\leq t}$, and $(2)$ for all $u,v\in V,t\geq 1$, the probability that $u$ and $v$ are in different sets of $\Ph_t$ is at most $w(u,v)/(t\cdot \poly(\log n))$. 
They proved that $\Ph_t$ is a suitable replacement for the set of connected components of $G_{\leq t}$.
That is, if we consider each set $C$ of $\Ph_t$ for $t\geq 1$, where $C$ is composed of sets $C_1,C_2,\ldots,C_r$ of $\Ph_{t/2}$, and find an arbitrary spanning tree in $G_{\leq t}$ connecting the super-nodes $C_1,\ldots,C_r$, then the union of all edges found by the above process gives an $O(1)$-approximate MST. 

The relaxed partitioning $\Ph_t$ does not increase the MST cost too much because the probability of cutting an MST edge $(u,v)$ at $\Ph_t$ is at most $w(u,v)/(t\cdot \poly(\log n))$, and cutting $(u,v)$ at $\Ph_t$ will introduce an additional cost of at most $t$.
Therefore, the overall expected overhead introduced by $(u,v)$ across all levels is at most $\sum_{t=1,2,4,\ldots,W} t\cdot w(u,v)/(t\cdot \poly(\log n)) = w(u,v)/ \poly(\log n)$. 
Therefore the problem becomes finding such a hierarchy of partitions $\Ph_{1/2},\Ph_1,\Ph_2,\ldots,\Ph_{W}$.

To find the hierarchical partitioning, \cite{jayaram2024massively} proposed the following 2-phase framework:
\begin{enumerate}
\item \textbf{Construct $\hat{P}_t$ for pivots $t=1/\alpha,1,\alpha,\alpha^2,\ldots,W$:} Here $\alpha$ is a sufficiently large $\poly(\log n)$. 
Firstly, an auxiliary hierarchy of partitions $\P'_{1/\alpha}, \P'_{2/\alpha}, \P'_{4/\alpha}\ldots, \P'_{W}$ is computed using randomly shifted hypergrids in Euclidean space. The diameter of each vertex set of $\P'_{t}$ is at most $t$, and the probability that two vertices $u, v$ are not in the same set of $\P'_{t}$ is at most $w(u,v)\cdot \poly(\log n)/t$.

Then, for each $t= 1,\alpha,\alpha^2,\ldots,W$, the partition $\Ph_t$ is obtained by computing the connected components of $G_{\leq t}$ after removing the edges that cross $\P'_{\alpha t}$.
To compute such connected components, one can start with contracting the vertices in each set of $\P'_t$ as a super-node since each set of $\P'_t$ has diameter at most $t$, and then compute connected components on the edge-removed $G_{\leq t}$.
The probability that an edge $(u,v)$ crosses $\Ph_t$ is equal to the probability that $(u,v)$ crosses by $\P'_{\alpha t}$ which is at most $w(u,v)\cdot \poly(\log n)/(\alpha t) = w(u,v)/(t\cdot \poly(\log n))$. This property is crucial to the analysis of the approximation ratio.
\item \textbf{In parallel for all $t=1/\alpha, 1, \alpha, \alpha^2,\ldots,W$, construct $\Ph_{2 t}, \Ph_{4 t}, \Ph_{8 t},\ldots, \Ph_{\alpha t / 2}$:}
The problem is solved in a divide-and-conquer manner.
In the first iteration, $\Ph_t$ and $\Ph_{\alpha t}$ are given and the goal is to compute $\Ph_{\sqrt{\alpha}  t}$.
One can start with contracting each vertex set in the auxiliary partition $\P'_{\sqrt{\alpha} t}$ as a super-node and compute connected components on the graph $G_{\leq \sqrt{\alpha}  t}$ after removing the edges that cross $\Ph_{\alpha t}$.
$\Ph_{\sqrt{\alpha}  t}$ is composed of the connected components found by the above process.
Then the problem requires recursively computing: $(1)$ $\Ph_{2 t}, \Ph_{4 t}, \ldots, \Ph_{\sqrt{\alpha}  t / 2}$ and $(2)$ $\Ph_{2\sqrt{\alpha} t}, \Ph_{4\sqrt{\alpha} t}, \ldots, \Ph_{\alpha t / 2}$.
\end{enumerate}

Two major challenges arise when implementing the above algorithm in the MPC model: (1)~computing the auxiliary hierarchy $\{\P'_i\}$, (2)~computing the connected components.
For the former, \cite{jayaram2024massively} solved it by applying randomly shifted grids over Euclidean points.
For the latter, they observed that each connected component subroutine starts with super-nodes (incomplete connected components) obtained by merging vertex sets of $\P'_t$ for some $t$, and all points within the same connected component are within distance $t\cdot \poly(\log n)$.
In this case, they showed that one only needs to reduce the number of incomplete connected components by a factor of $\poly(\log n)$, and connecting the remaining incomplete connected components in an arbitrary way that does not increase the final MST cost by too much.
By applying standard leader compression based connected components algorithms~\cite{andoni2018parallel,behnezhad2019near,liu2020connected,karloff2010model,reif1984optimal}, each connected component subroutine only takes $O(\log \log n)$ rounds.
Since the second phase has $O(\log\log \log n)$ levels of recursions, the overall number of rounds is $O(\log\log(n)\cdot \log\log\log(n))$.

\vspace{-0.2cm}
\subsubsection*{Beyond Shifted Grids: Our Algorithm for General Metrics}

A major challenge for general metrics is to compute the auxiliary hierarchy of partitions. While \cite{jayaram2024massively} use randomly shifted hypergrids in the Euclidean space, there is no natural data-oblivious hierarchical partitioning of general metrics. To get around this, we carefully implement a low-diameter decomposition of \citet*{MillerPX13} (henceforth MPX) in our setting. 

Another issue with the algorithm in \cite{jayaram2024massively} is the second phase. It is inefficient in terms of both the round-complexity and the approximation ratio. That is, it incurs an additional $\log \log \log n$ factor on the running-time, and it results in a large constant approximation instead of our desired $1+\epsilon$ approximation ratio. In the following, we first explain how to handle the first phase for when the input is a general metric, and then describe our novel second-phase algorithm.

Let us consider the complete metric graph $G$ over $V$, i.e., each edge $(u,v)$ has weight $w(u,v)$.
To start, we utilize the low-diameter decomposition~\cite{MillerPX13} of $G$.
For a given parameter $t$, it can be computed as follows: 
(1) Add an additional source node, and add an edge from the source to each vertex in $G$. The weight of each edge is drawn from the exponential distribution with mean $t / \log n$. 
(2)~Compute a single-source shortest-path tree from the source node.
The vertex sets in the decomposition are the subtrees directly below the root.
Since the weights in $G$ come from a metric, the depth of this tree is always at most $2$, and thus it can be easily computed in $O(1)$ rounds of the MPC model (e.g.\ using the Bellman-Ford algorithm).

We compute an MPX decomposition for each $t = \alpha^k$, where $\alpha = \Theta \left(\frac{\log^2 n}{\epsilon}\right)$, up to $t = W$.
As a result, we get a sequence of partitions $\{\P^0_{\alpha^k}\}_k$ with similar guarantees to those of the shifted grids. That is, on level $t$, each vertex set has diameter at most $t$, and the crossing probability of an edge $(u,v)$ is $w(u,v)\cdot \poly(\log n)/ t$.

We need to make sure the partitions form a hierarchy.
To achieve this, we construct 
$\P^\MPX=\{\P^\MPX_{\alpha^k}\}_k$ as follows. For each level $t = \alpha^k$, the partition $\P^\MPX_t$ is obtained by intersecting all the decompositions $\P^0_{t'}$ for $t'\geq t$, i.e., $u$ and $v$ are in the same set of $\P^\MPX_t$ iff $u$ and $v$ are in the same set of $\P^0_{t'}$ for all $t'\geq t$.
Notice that this does not increase the diameter of each set.
In addition, the crossing probability of an edge $(u,v)$ is still bounded by the geometric sum $\sum_{t'\geq t} w(u,v)\cdot \poly(\log n) /t'$ which is at most $\frac{1}{1-\frac{1}{\alpha}} \leq 2$ times the original probability.
We need to further process $\P^\MPX$ in order to lower the crossing probabilities by a factor of $\alpha$, and obtain the final hierarchy $\Ph = \{\Ph_{\alpha^k}\}_{k}$.

\begin{wrapfigure}[12]{r}{0.4\textwidth}
    \vspace{-20pt}
    \centering
    \resizebox{0.4\textwidth}{!}{%
    \begin{tikzpicture}
    \node[draw, circle, minimum size=8cm, label=above:$\mathcal{P}^{\text{MPX}}_{\alpha t}$] (bigcluster) {};
    
    \path (bigcluster) ++(180:3cm) coordinate (P1);
    \path (bigcluster) ++(108:3cm) coordinate (P2);
    \path (bigcluster) ++(36:3cm) coordinate (P3);
    \path (bigcluster) ++(-36:3cm) coordinate (P4);
    \path (bigcluster) ++(-108:3cm) coordinate (P5);
    
    \fill[cyan, fill opacity=0.2] plot [smooth cycle, tension=1] coordinates {($(P1) + (-0.8,0)$) ($(P2) + (0.5,0.6)$) ($(P5) + (0.5,-0.6)$)};
    \draw[dashed] plot [smooth cycle, tension=1] coordinates {($(P1) + (-0.8,0)$) ($(P2) + (0.5,0.6)$) ($(P5) + (0.5,-0.6)$)};
    
    \fill[cyan, fill opacity=0.2] plot [smooth cycle, tension=1] coordinates {($(P3) + (0.9,0.1)$) ($(P3) + (-0.9,0.1)$) ($(P4) + (-0.9,-0.1)$) ($(P4) + (0.9,-0.1)$)};
    \draw[dashed] plot [smooth cycle, tension=1] coordinates {($(P3) + (0.9,0.1)$) ($(P3) + (-0.9,0.1)$) ($(P4) + (-0.9,-0.1)$) ($(P4) + (0.9,-0.1)$)};

    \draw (P2) -- (P1) node[midway, above, sloped] {$\leq t$};
    \draw (P2) -- (P1); 
    \draw (P1) -- (P5);
    \draw (P3) -- (P4);

    \draw (P2) -- ($(P2) + (-1.6,1.5)$);

    \draw (P1) -- ($(P2) + (-4,-2.5)$);
    
    \draw (P3) -- ($(P3) + (+2.3,1.3)$);
    
    \node[draw, circle, minimum size=1cm, fill=white] at (P1) {$\mathcal{P}^{\text{MPX}}_{t}$};
    \node[draw, circle, minimum size=1cm, fill=white] at (P2) {$\mathcal{P}^{\text{MPX}}_{t}$};
    \node[draw, circle, minimum size=1cm, fill=white] at (P3) {$\mathcal{P}^{\text{MPX}}_{t}$};
    \node[draw, circle, minimum size=1cm, fill=white] at (P4) {$\mathcal{P}^{\text{MPX}}_{t}$};
    \node[draw, circle, minimum size=1cm, fill=white] at (P5) {$\mathcal{P}^{\text{MPX}}_{t}$};

    \node at ($(P1) + (2.8,0)$) {$\widehat{\mathcal{P}}_t$};
    
    \node at ($(P1) + (6.2,0)$) {$\widehat{\mathcal{P}}_t$};

\end{tikzpicture}
    }
\end{wrapfigure}

The final hierarchy $\Ph = \{\Ph_{\alpha^k}\}_k$ is obtained by extending the partitions of $\P^\MPX$.
The partition $\Ph_t$ is meant to approximate the set of connected components with respect to the edges of weight at most $t$ that do not cross $\P^\MPX_{\alpha t}$, call these edges $G_{\leq t}[\P^\MPX_{\alpha t}]$ (see the figure on the right).
For every level $t = \alpha^k$, $\Ph_t$ is obtained by executing $\Theta\left(\log \frac{1}{\epsilon} + \log \log n\right)$ rounds of leader compression on $\P^\MPX_t$, using the edges of $G_{\leq t}[\P^\MPX_{\alpha t}]$.
After the leader compression,
the \enquote{incomplete} components within each vertex set of $\P^\MPX_{\alpha t}$ are joined together, where we say a component is incomplete when it still has outgoing edges that would have been used if we had not stopped the leader compression early.

For the second phase, we develop a completely new algorithm.
The original phase-2 algorithm aimed to generate intermediary partitions $\Ph_{2\cdot t},\Ph_{4\cdot t},\Ph_{8\cdot t},\ldots,\Ph_{\alpha\cdot t /2}$, between $\Ph_t$ and $\Ph_{\alpha t}$.
The final approximate MST was obtained by computing an arbitrary spanning tree for each level.
There are two major disadvantages to this: (1) Generating all intermediate partitions $\Ph_{t'} $ introduces an additional $O(\log\log\log n)$-round overhead.
(2) Computing an arbitrary spanning tree for each level introduces a factor of $2$ in the approximation ratio.
To address these two issues, we forgo generating the intermediary partitions.
Instead, we use a novel variation \Boruvka's algorithm to build the approximate MST.

For each level $t = \alpha^k$, we regard each vertex set of $\Ph_{t/\alpha}$ as a super-node and try to find a minimum spanning forest over the edges of weight at most $\alpha t$ that do not cross the next level partition $\Ph_{t}$, call these edges $G_{\leq \alpha t}[\Ph_t]$.
That is, the forest connects all the super-nodes inside a set of $\Ph_t$.
This way, we avoid computing the intermediate levels, and it saves us the $O(\log\log\log n)$ factor in the number of rounds.
We propose the following variant of \Boruvka's algorithm:
\begin{enumerate}
\item For each super-node in the graph, we sample it as a leader with probability $1/2$.
\item For each non-leader super-node $u$, we look at its nearest neighbor super-node $v$ in $G_{\leq \alpha t}[\Ph_t]$.
If $v$ is a leader, we merge $u$ into $v$, i.e., a new super-node includes both members of super-nodes $u$ and $v$.
\item Repeat the above steps for $T$ rounds.
\end{enumerate}

Had we run this algorithm for $\Theta(\log n)$ rounds, the exact minimum forest with respect to $G_{\leq \alpha t}[\Ph_t]$ would have been found.
As it turns out, using many fewer rounds, i.e.\ $T = \Theta\left(\log \frac{\alpha}{\epsilon}\right)$ computes most of the minimum spanning forest, and hence suffices for a $(1 + \epsilon)$-approximation.
Consider level $t=\alpha^k$. Let $C$ be an arbitrary vertex set of $\Ph_{t}$ that is composed of sets $C_1,C_2,\ldots,C_r$ of $\Ph_{t/\alpha}$.
Let us apply the above algorithm to super-nodes $C_1,\ldots C_r$ over the edges $G_{\leq \alpha t}[\Ph_t]$.
At any point, unless $C_1,\ldots,C_r$ are merged into one super-node, during the next round, each super-node joins another one with a constant probability.
Therefore, after $T$ rounds, we expect the number of unmerged super-nodes in $C$ to be (roughly) $r/2^{\Theta(T)}$.
That is, only a $\left(\poly\left(\frac{\epsilon}{\alpha}\right)\right)$-fraction of the minimum spanning forest edges remain undiscovered after $T$ rounds.
We connect the remaining super-nodes in $C$ using arbitrary edges of weight at most $\alpha t$. This is fine, as we have assured there are not many super-nodes left inside $C$.

Observe that the approximation ratio is controlled by $\alpha$.
A larger $\alpha$ improves the approximation ratio by introducing lower crossing probabilities for the edges,
whereas a smaller $\alpha$ means more parallelization.
Specifically, (approximately) computing the minimum spanning forest for $\Ph_{t/\alpha}$ requires $\Omega\left(\log \frac{\alpha}{\epsilon}\right)$ rounds.
To achieve a $(1 + \epsilon)$-approximation ratio, we need to set $\alpha = \Theta\left(\frac{\log^2 n}{\epsilon}\right)$, which results in an $O\left(\log \frac{1}{\epsilon} + \log \log n\right)$-round algorithm.
\subsection{Technical Overview of \cref{thm:lb}}\label{sec:tech-lb}

One of our main contributions is a $\Omega\left(\log \frac{1}{\epsilon}\right)$-round lower bound for computing a $(1 + \epsilon)$-approximate of the MST in $(1, 2)$-metrics (\cref{thm:lb}).
In fact, our lower bound even applies to the $(1, 2)$-metrics where the weight-$1$ edges form disjoint cycles. We consider only such metrics for this discussion.
An approximate MST excludes at least one weight-$1$ edge from each cycle, and (on average) not many more. 
Specifically, if the weight-$1$ edges form two cycles, one of length $\frac{1}{10\epsilon}$, and one of length $n - \frac{1}{10\epsilon}$,
then a $(1 + \epsilon)$-approximate MST excludes at least one edge from each cycle, and excludes at most a total of $2\epsilon n$ edges.

Intuitively, to find a $(1+\epsilon)$-approximation, every vertex needs to know its $\Omega\left(\frac{1}{\epsilon}\right)$-neighborhood, i.e.\ the vertices that are within $\frac{1}{\epsilon}$ hops with respect to weight-$1$ edges. 
This is because (roughly) at most one in every $\frac{1}{\epsilon}$ edges can be excluded, but at the same time we need to make sure at least one edge from each cycle is removed.
Specifically, in the described example, knowledge of the $\frac{1}{10\epsilon}$-neighborhood is required to detect whether the vertex is in the smaller cycle.
$(1 + \epsilon)$-approximation is not possible otherwise since the edges in the small cycle must be excluded with a high rate of at least $10\epsilon$, whereas the rest of the edges must be excluded with a rate of at most $2\epsilon$.

This intuition can be formalized to obtain a $\Omega\left(\frac{1}{\epsilon}\right)$ lower bound in the distributed LOCAL model where the nodes are allowed to communicate over the weight-$1$ edges of the metric (see \cref{subsec:lb-local-random}).
Then, using the framework of \cite{GhaffariKU19} and \cite{CzumajDP21a}
it can be lifted to obtain a (conditional) $\Omega\left(\log \frac{1}{\epsilon}\right)$ lower bound for \emph{component-stable} MPC algorithms, where the components are with respect to weight-$1$ edges (see \cref{sec:lift-to-MPC}).
That is, the output of the MPC algorithm on a connected component of the weight-$1$ edges is required to be independent of the rest of the graph.
While many MPC algorithms are component-stable, many natural ones are not, as discussed in \cite{CzumajDP21a}.
We show that the assumption on component stability can be dropped. Specifically, in \cref{thm:lb},
we give a lower bound assuming solely the \onetwocycle{} conjecture.

To prove the lower bound, we show a reduction that given
an $R$-round algorithm for MST approximation, produces an
$O(\log_{1/\epsilon}n \cdot R)$-round algorithm for solving the \onetwocycle{} problem. 
This implies that a $(1 + \epsilon)$-approximation of MST in $(1, 2)$-metrics is not possible in $o\left(\log \frac{1}{\epsilon}\right)$ rounds under the \onetwocycle{} conjecture.
Given the input, the reduction repeatedly decreases the length of the cycles by a factor of $\epsilon$ to verify if it is a cycle of length $n$.

To remove the requirement of component stability, we present a way to transform MPC algorithms using a random reordering of the vertex IDs.
That is, given an MPC algorithm $\A$, and some input graph $G$, we apply a random permutation to the vertex IDs of $G$ and feed the result to $\A$.
The resulting algorithm $\A'$ has a key property.
For any two isomorphic vertices/edges,
the output of $\A'$ for them has the same distribution.

By applying this transformation to an approximate MST algorithm for $(1, 2)$-metrics where the weight-$1$ edges form cycles, we get an algorithm such that for any two weight-$1$ edges that are in cycles of the same length, the probability of being excluded is the same. We use these probabilities in a subroutine to detect the $\Theta\left(\frac{1}{\epsilon}\right)$-neighborhood of a vertex.
The subroutine is then utilized in the reduction that implies the lower bound.
That is, given a cycle, the subroutine is used to decrease the length by a factor of $\epsilon$.

The $\frac{1}{10\epsilon}$-neighborhood of a vertex is detected as follows.
Assume for now that the input graph is a cycle of length $n$.
Given two vertices $u_1$ and $u_2$, the goal is to assert whether they are at most $\frac{1}{10\epsilon}$ hops away in the cycle.
We start by picking one adjacent edge for each of $u_1$ and $u_2$, and doing a two-switch (see \cref{fig:two-switch-simple}).
If $u_1$ and $u_2$ are close, and the appropriate edges are picked,
this results in the cycle being broken down into two cycles: one of length at most $\frac{1}{10\epsilon}$ containing $u_1$ and $u_2$, and another of length at least $n - \frac{1}{10\epsilon}$ containing.
We consider the $(1, 2)$-metric where the weight-$1$ edges are these two cycles, and study the probabilities of each edge being excluded when $\A'$ is run on the metric.
It can be shown that the edges in the smaller cycle are excluded with probability at least $10\epsilon$, and the edges in the larger cycle with probability at most $2\epsilon$.
Therefore, they can be identified, and we can assert that $u_1$ and $u_2$ are close.

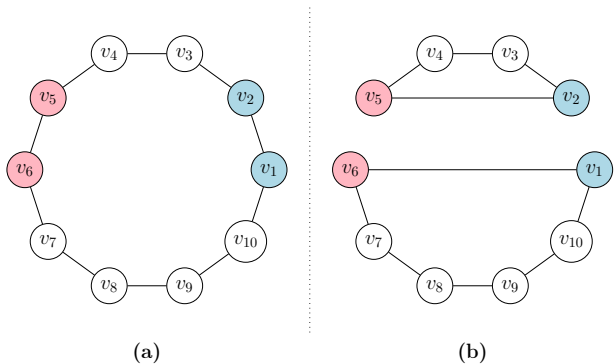
\begin{figure}[H]
    \centering
    \resizebox{0.5\textwidth}{!}{%
    \begin{tikzpicture}[scale=0.72, transform shape, every node/.style={font=\Large}]
    \definecolor{lightblue}{RGB}{173, 216, 230}
    \definecolor{lightred}{RGB}{255, 182, 193}

    \foreach \i in {1,...,10}
    {
        \ifthenelse{\i=1 \OR \i=2}
            { \ifthenelse{\i=2}{
            \node[draw, circle, fill=lightblue, minimum size=1.75em] (v\i) at ({360/10 * (\i - 1)}:3) {\Large $u'$};
            }{
            \node[draw, circle, fill=lightblue, minimum size=1.75em] (v\i) at ({360/10 * (\i - 1)}:3) {\Large $v'$};
            }
            }
            {\ifthenelse{\i=5 \OR \i=6}
                {
                \ifthenelse{\i=5}{
                \node[draw, circle, fill=lightred, minimum size=1.75em] (v\i) at ({360/10 * (\i - 1)}:3) {\Large $u$};
                }{
                \node[draw, circle, fill=lightred, minimum size=1.75em] (v\i) at ({360/10 * (\i - 1)}:3) {\Large $v$};
                }
                }
                {\node[draw, circle, minimum size=1.75em] (v\i) at ({360/10 * (\i - 1)}:3) {};}
            }
    }
    
    \foreach \i in {1,...,9}
    {
        \pgfmathtruncatemacro\next{\i + 1}
        \draw (v\i) -- (v\next);
    }
    \draw (v10) -- (v1);
    
    \node at (0, -4.5) {\Large \textbf{(a)}};
    
    \draw[dotted] (4, -4) -- (4, 4);
    
    \begin{scope}[xshift=8cm]
        \foreach \i in {1,...,10}
        {
            \ifthenelse{\i=1 \OR \i=2}
            { \ifthenelse{\i=2}{
            \node[draw, circle, fill=lightblue, minimum size=1.75em] (v\i) at ({360/10 * (\i - 1)}:3) {\Large $u'$};
            }{
            \node[draw, circle, fill=lightblue, minimum size=1.75em] (v\i) at ({360/10 * (\i - 1)}:3) {\Large $v'$};
            }
            }
            {\ifthenelse{\i=5 \OR \i=6}
                {
                \ifthenelse{\i=5}{
                \node[draw, circle, fill=lightred, minimum size=1.75em] (v\i) at ({360/10 * (\i - 1)}:3) {\Large $u$};
                }{
                \node[draw, circle, fill=lightred, minimum size=1.75em] (v\i) at ({360/10 * (\i - 1)}:3) {\Large $v$};
                }
                }
                {\node[draw, circle, minimum size=1.75em] (v\i) at ({360/10 * (\i - 1)}:3) {};}
            }
        }
        
        \foreach \i in {1,...,9}
        {
            \pgfmathtruncatemacro\next{\i + 1}
            \ifthenelse{\i=1 \OR \i=5}
                {}
                {\draw (v\i) -- (v\next);}
        }
        \draw (v10) -- (v1);
        
        \draw (v1) -- (v6);
        \draw (v2) -- (v5);
        
        \node at (0, -4.5) {\Large \textbf{(b)}};
    \end{scope}
\end{tikzpicture}
    }
    \caption{A two-switch on edges $(u, v)$ and $(u', v')$, replaces them with $(u, u')$ and $(v, v')$.
    Here, the two-switch breaks the cycle into two cycles. See also \cref{def:two-switch} and \cref{clm:two-switch}.}
    \label{fig:two-switch-simple}
\end{figure}

Using the $\frac{1}{10\epsilon}$-neighborhoods, we aim to decrease the length of the cycle by a factor of $\epsilon$. After doing this $O(\log_{1/\epsilon} n)$ times, we are able to verify whether the original graph was a cycle of length $n$.
A potential approach is to select certain vertices as leaders and contract their $\frac{1}{10\epsilon}$ neighborhoods.
It is not immediately clear this is possible in $O(1)$ rounds,
since we need at most $O(\epsilon n)$ leaders such that each vertex joins a leader, and the vertices that join the same leader form a connected subgraph (i.e.\ a subpath) of the cycle.
Instead of contracting, we break down the cycle into $\frac{1}{\epsilon}$ smaller cycles by removing all the edges, and connecting each vertex to the two vertices that were exactly $\frac{1}{\epsilon}$ hops away in the cycle (\cref{fig:cycle-dividing}). This can be done repeatedly. However, the subroutine for computing the $\frac{1}{10\epsilon}$-neighborhoods becomes more complicated.

\section{Preliminaries}\label{sec:preliminaries}

\textbf{The MPC Model:}
In MPC model, an input of size $N$ is distributed among $M$ machines with local memory $S$ and infinite computing power.
The number of machines is ideally set to $M = O(N/S)$, or slightly higher otherwise.
The computation takes place in synchronous rounds.
Each round has a processing phase,
where each machine makes operations on its local data,
and a communication phase,
where each machine sends messages to some of the other machines.
The total size of the messages sent or received by each machine should not exceed $S$.

There are different regimes for the choice of $S$. We study the \emph{strictly sublinear} regime (a.k.a.\ fully scalable algorithms).
That is, given a graph of $n$ vertices, each machine has local space $n^\delta$, where $0 < \delta < 1$.

\textbf{The Input:}
The input of our algorithm is a metric over $n$ points, expressed as $n^2$ words. We refer to this metric as a complete graph $G$, where the vertices represent the points,
and the weight of the edge between two vertices $u$ and $v$ is equal to their distance in the metric.

\subsection*{Graph Notation}
Given a complete weighted graph $G$, we use $w(u, v)$ to denote the weight of the edge between $u$ and $v$, and $W = \poly(n)$ to denote the maximum weight.
We use $G_{\leq t}$ to denote the subgraph of $G$ consisting of all the edges of weight at most $t$, and use $\C_t$ to denote the set of connected components in $G_{\leq t}$. 
We use $\MST(G)$ to denote its minimum spanning tree, and $w(\MST(G))$ to denote its weight.
Throughout the paper, we break ties based on edge ID,
so that the minimum spanning tree is unique.

Given a graph $G$, a partition $\P$ of the vertices $V(G)$ is a family of vertex sets such that each vertex is a member of exactly one set.
A partition $\P_1$ \emph{refines} a partition $\P_2$ if for every set $A \in \P_1$, there exists a set $B \in \P_2$ such that $A \subseteq B$.
We say an edge $(u, v)$ crosses the partition $\P_1$ if $u$ and $v$ are in different sets of $\P_1$.
We say a partition $\P_1$ captures an edge $e$ if $e$ does not cross $\P_1$.
We use $G[\P_1]$ to denote the subgraph consisting of the edges that do not cross $\P_1$.

\begin{definition}
    Given a graph $G$, a \emph{hierarchy of vertex partitions} $\sP$ is a sequence of vertex partitions $\{\P_i\}_{i=0}^k$, such that 
    \begin{enumerate} 
        \item each partition is a refinement of the next, and
        \item the last partition groups all the vertices together, meaning $\P_k = \{ V(G) \}$.
    \end{enumerate}
\end{definition}

For our purposes, we index the hierarchy by a geometric series corresponding to the weights. The $(k+1)$-th partition is denoted by $\P_{\alpha^k}$, where $\alpha \geq 1$ is a parameter and $k \geq 0$. The last partition is $\P_W$. The hierarchy is denoted by $\sP=\{\P_{\alpha^k}\}_{k \geq 0}$.
We use $L = \log_\alpha W$ to denote the number of levels.
For a level $t = \alpha^k$, $\P_{\alpha t}$ is the next partition of $\P_t$, and $\P_{t/\alpha}$ is the previous partition.
If $t = 1$, we use $\P_{t/\alpha}$ to simply denote the trivial partition where each vertex is in its own set, i.e.\ each set is a singleton.

Given a partition $\P$ and a set of edges $E$, we use $\P \oplus E$ to denote the partition obtained from joining the sets of $\P$ using the edges in $E$.
That is, two vertices $s$ and $t$ are in the same set of $\P \oplus E$ if there exists a \enquote{path} of edges $(u_0, v_0), (u_1, v_1), \ldots, (u_k, v_k) \in E$, where $s$ and $u_0$ are in the same set of $\P$, $v_i$ and $u_{i+1}$ are in the same set of $\P$ for all $i$,  and $v_k$ and $t$ are in the same set of $\P$.
Equivalently, one could define $\P \oplus E$ as the finest partition that is refined by both $\P$ and the set of connected components of $E$.

\begin{definition} \label{def:respecting-tree}
    Given a weighted graph $G$, a hierarchy $\sP=\{\P_{\alpha^k}\}_{k \geq 0}$, and a sequence of edge sets $\{F_{\alpha^k}\}_k$, a spanning tree $T$ \emph{respects} $(F_{\alpha^k} \mid \P_{\alpha^k})$ if it can be expressed as the disjoint union of a sequence of edge sets $\{E_{\alpha^k}\}_{k \geq 0}$, such $E_{\alpha^k}\subseteq F_{\alpha^k}$, and $\P_{\alpha^k} = \P_{\alpha^{k-1}} \oplus E_{\alpha^k}$, i.e.\ the set of connected components with respect to $\cup_{i \leq k}E_{\alpha^i}$ is $\P_{\alpha^k}$.

    We use $\MST(F_{\alpha^k} \mid \P_{\alpha^k})$ to denote the minimum spanning tree that respects $(F_{\alpha^k} \mid \P_{\alpha^k})$.
\end{definition}

Intuitively, starting with the empty graph, for each level $k$, the tree $T$ uses some edges of $F_{\alpha^k}$ (i.e.\ $E_{\alpha^k}$) to join components of $\P_{\alpha^{k-1}}$, such that all the vertices within a set of $\P_{\alpha^k}$ are connected. 

\subsection*{MPX Decomposition}
A key ingredient of our algorithm is a low-diameter decomposition of graphs using random shifts due to \citet{MillerPX13}.
Given any weighted graph $G$, this is obtained by drawing, for each vertex $u$, a \emph{delay} $\delta_u$ from the exponential distribution with expectation $\Theta(\frac{t}{\log n})$, i.e.\ the distribution with density function $f(x) = c\frac{\log n}{t} \exp\left(-c\frac{\log n}{t} x\right)$.
Then, the center for $c(u)$ for a vertex $u$ is defined as the vertex that minimizes the \emph{shifted distance}. That is, $c(u) = \argmin_v d(u, v) - \delta(v)$, where $d(u, v)$ is the distance between $u$ and $v$ (which is equal to $w(u, v)$ in case of metrics). This results in a decomposition, hereafter referred to as an MPX decomposition with diameter $t$, comprised of sets of vertices with the same center.
In their paper, they show:
\begin{proposition}[\cite{MillerPX13}] \label{prp:mpx}
    Let $\P_t$ be an MPX decomposition with parameter $t$. It holds that:
    \begin{enumerate}
        \item any set in $\P_t$ has weighted diameter at most $t$, with high probability, and
        \item for any edge $(u, v)$, the probability that $u$ and $v$ are in different sets of $\P_t$ is $O\left(  \frac{ w(u, v) \log n}{t}\right)$.
    \end{enumerate}
\end{proposition}
\section{An $O\left(\log \frac{1}{\epsilon} + \log \log n\right)$ Algorithm for Metric MST}

\subsection{The Algorithm}
First, we give a brief description of \Cref{alg:main}.
Given a complete weighted graph $G$ such that the weights satisfy the triangle inequality, let $\alpha = \Theta\left(\frac{\log^2 n}{\epsilon}\right)$. To begin, the algorithm creates a hierarchy.
The levels in the hierarchies are always indexed by $\{1, \alpha, \alpha^2, \ldots, W\}$. Ideally, the algorithm would have used $\sC = \{\C_{\alpha^k}\}_{k\geq0}$ as its hierarchy, where $\C_{\alpha^k}$ is the set of connected components in $G_{\leq \alpha^k}$.
As that is not possible under the \onetwocycle{} conjecture, the algorithm uses an approximation of $\sC$.

Initially, a hierarchy $\sP^\MPX$ is obtained by intersecting MPX decompositions with different parameters.
This hierarchy satisfies the same diameter and crossing probability guarantees as in \cref{prp:mpx}.
Another hierarchy $\sP$ is defined by extending $\sP^\MPX$. For each level $t=\alpha^k$, the partition $\P_t$ is built by joining the sets in $\P^\MPX_t$ using edges of weight at most $t$ that do not cross $\P^\MPX_{\alpha t}$. In short, $\P_t$ is the set of connected components in $G_{\leq t}[\P^\MPX_{\alpha t}]$.
This hierarchy is a suitable approximation of $\sC$;
however, it cannot be computed efficiently either. 
Therefore, the algorithm uses a hierarchy $\sPh$ as an approximation to $\sP$.

For each level $t$, the partition $\Ph_t$ is constructed by running $\Theta\left(\log \frac{1}{\epsilon} + \log \log n\right)$ rounds of leader compression on $\P^\MPX_t$ using the edges of $G_{\leq t}[\P^\MPX_{\alpha t}]$.
After the leader compression subroutine has halted,
all the \emph{incomplete} components within each set of $\P^\MPX_{\alpha t}$ are joined together.
Here, incomplete means that the component still has outgoing edges of $G_{\leq t}[\P^\MPX_{\alpha t}]$,
and if the leader compression had not stopped early,
it would have been connected to some other component.
Intuitively, $r = \Theta(\log \frac{1}{\epsilon} + \log \log n)$ rounds of leader compression constructs $\P_t$ up to an \enquote{error} of only $2^{-\Omega(r)}=\poly\left(\frac{\epsilon}{\log n}\right)$.

Finally, the algorithm builds a tree using the hierarchy $\sPh$. For each level $t$, it runs \linebreak $\Theta\left(\log \frac{1}{\epsilon} + \log \log n\right)$ rounds of the modified \Boruvka{} algorithm on $\Ph_{t/\alpha}$ using the edges $G_{\leq \alpha t}[\Ph_{t}]$.\footnote{Instead, one could have used $G_{\leq t}[\Ph_t]$ here, and it may be more instructive to think of that algorithm. $G_{\leq \alpha t}[\Ph_t]$ is used for the simplicity of the proof.}
Upon the termination of this subroutine,
within each set of $\Ph_{t}$, all the remaining sets are connected together arbitrarily.
The modified \Boruvka{} algorithm is a combination of \Boruvka's algorithm and a coin-flipping scheme similar to the one in leader compression.

\begin{algorithm}[H]
\caption{An $O\left(\log \frac{1}{\epsilon} + \log \log n\right)$ Algorithm for Metric MST}
    \label{alg:main}

    \textbf{input: } A complete weighted graph $G$, and a parameter $\epsilon > 0$.

    Let $\alpha = \Theta\left(\frac{\log^2 n}{\epsilon}\right)$.

    Compute MPX decomposition $\P^0$ with diameter $t$, for each level $t = \alpha^k$.
    
    Derive a hierarchy of partitions $\sP^\MPX$, by taking intersections of $\P^0$'s, making them nested.

    Obtain the final hierarchy $\sPh$,
    where $\Ph_t$ is obtained by running $\Theta\left(\log \frac{1}{\epsilon} + \log \log n\right)$ rounds of leader compression on $\P^\MPX_t$ using the edges of $G_{\leq t}[\P^\MPX_{\alpha t}]$, and connecting the incomplete components within each set of $\P^\MPX_{\alpha t}$.
    \label{step:mst-leader-compression}

    \For{each level $t = \alpha^k$}{ \label{step:mst-loop-start}
        Run $\Theta\left(\log \frac{1}{\epsilon} + \log \log n\right)$ rounds of the modified \Boruvka{} algorithm (\Cref{alg:boruvka}), with initial components $\Ph_{t / \alpha}$,
        and edges $G_{\leq \alpha t}[\Ph_t]$, call the output edges $E^1_t$.
        \label{step:mst-boruvka}

        Connect any remaining components within $\Ph_t$ arbitrarily, call the used edges to $E^2_t$.
        \label{step:mst-join-arbitrarily}

        Let $E_t = E^1_t \cup E^2_t$.
    \label{step:mst-loop-end}
    }

    \Return{$T = \bigcup_k E_{\alpha^k}$.}
    
\end{algorithm}

Now we give a more detailed description of the algorithm.
The parameter $\alpha$ is chosen such that 
such that $\alpha \geq \frac{L \log n}{\epsilon}$ where $L = \log_\alpha n$ is the number of levels, e.g.\ take $\alpha = \Theta\left(\frac{\log^2 n}{\epsilon}\right)$. 
Creating the MPX decomposition $\{\P^0_t\}_{t = \alpha^k}$ can be done independently for each level in $O(1)$ rounds,
where each set in $\P^0_t$ has diameter $t$ with high probability, and an edge of weight $w$ crosses $\P^0_t$ with probability $O\left(\frac{w \log n}{t}\right)$ (\cref{prp:mpx}).
Note that the sequence MPX decompositions $\sP^0 = \{\P^0_{\alpha^k}\}_{k \geq 0}$ is not a hierarchy, which is necessary for our application.

The MPX hierarchy $\sP^\MPX = \{\P^\MPX_{\alpha^k}\}_{k\geq 0}$ is built by intersecting the decompositions.
Formally, we let two vertices be in the same set of $\P^\MPX_{\alpha^k}$, if and only if they are in the same set of $\P^0_{\alpha^i}$ for all $i \geq k$. Observe, that by intersecting the partitions of the next levels, we ensure that each partition is a refinement of the next. It can be seen that the diameter and the crossing probability guarantees still hold for $\sP^\MPX$ (\cref{clm:Pmpx-hierarchy}).

The final hierarchy $\sPh$ is created as follows.
For each level $t = \alpha^k$, the partition $\Ph_t$
is obtained by running $\Theta\left(\log \frac{1}{\epsilon} + \log \log n\right)$ rounds of leader compression on $\P^\MPX_t$ using the edges of $G_{\leq t}[\P^\MPX_{\alpha t}]$.
In each round,
every component $C$ flips a fair coin $x_C \in \{0, 1\}$ independently.
If $x_C = 0$ and $C$ has an edge (in $G_{\leq t}[\P^\MPX_{\alpha t}]$) to a component $C'$ with $x_{C'} = 1$, then $C$ joins $C'$ (choosing arbitrarily if more than one option is available).
After the leader compression,
all the \emph{incomplete} components within each set of $\P^\MPX_{\alpha t}$ are joined together.
Where a component is said to be {incomplete} if it still has outgoing edges of $G_{\leq t}[\P^\MPX_{\alpha t}]$.
Each round of leader compression, and the final joining, can be implemented in $O(1)$ rounds of MPC.
Note that this is done for each $t$ in parallel independently.

To construct the tree, 
the algorithm runs $r = \Theta\left(\log \frac{1}{\epsilon} + \log \log n\right)$ rounds of the modified \Boruvka{} algorithm on $\Ph_{t/\alpha}$ using the edges of $G_{\leq \alpha t}[\Ph_{t}]$.
In each round, each component $C$ flips a fair coin $x_C \in \{0, 1\}$ and finds its minimum outgoing edge $e_C \in G_{\leq \alpha t}[\Ph_{t}]$ (breaking ties based on ID). Let $C'$ be the component at the other end of $e_C$. If $x_C = 0$ and $x_{C'} = 1$, then $e_C$ is added to $E^1_t$, and $C$ joins $C'$. The number of rounds $r$ is chosen such that $2^r$ dominates $\frac{\alpha^3}{\epsilon}$. 

\begin{algorithm}
\caption{Modified \Boruvka's Algorithm}
\label{alg:boruvka}
\textbf{Input:} A weighted graph $G$, an initial partition $\Ph_{t/\alpha}$, a set of edges $E$, and the number of rounds $r$.

Let the set of components $\C := \Ph_{t/\alpha}$

Let $F = \emptyset$

\For{$i \in \{1, 2, \ldots, r\}$}{
    Flip a fair coin $x_C \in \{0, 1\}$ for each component $C \in \C$.
    
    \For{each component $C \in \C$}{
        Let $e_C \in E$ be the shortest outgoing edge of $C$.
        
        Let $C'$ be the component on the other endpoint of $e_C$.
        
        \If{$x_C = 0$ and $x_{C'} = 1$}{
            Add $e_C$ to the output $F$.
            
            Join $C$ to $C'$.
        }
    }
}
\Return{$F$}
\end{algorithm}

To finalize, the algorithm makes sure $\Ph_t = \Ph_{t/\alpha} \oplus E_t$ by arbitrarily connecting any remaining components within a set of $\Ph_t$.
Formally, let $\mathfrak{C}$ be the set of components \Cref{alg:boruvka} halts on. For each set of vertices $A \in \Ph_t$, let $\mathfrak{C}_A$ be the components of $\mathfrak{C}$ inside $A$ (note that $\mathfrak{C}$ is a refinement of $\Ph_t$).
Take an arbitrary component $C^* \in \mathfrak{C}_A$, and for any other component $C \in \mathfrak{C}_A$, add one of the edges between $C^*$ and $C$ to $E^2_t$.

The remainder of this section is devoted to the proof of \cref{thm:main}. First, we analyze the approximation ratio (\cref{lem:approx}).
Then, we go over the runtime analysis and the implementation details (\cref{lem:implementation}).

\subsection{Approximation Ratio}

We prove the following:

\begin{lemma} \label{lem:approx}
    \Cref{alg:main} computes a $(1 + \epsilon)$-approximation of the minimum spanning tree, in expectation.
\end{lemma}

In essence, the proof consists of three parts: $(i)$ the hierarchy $\sP$ is a good approximation for $\sC$, $(ii)$ the hierarchy $\sPh$ is a good approximation for $\sP$ and hence, also for $\sC$, and $(iii)$ the algorithm finds a $(1+\epsilon)$-approximate MST with respect to $\sPh$.
First, we prove that $\sP^\MPX$ is a hierarchy and examine its properties.

\begin{claim} \label{clm:Pmpx-hierarchy}
        The sequence of partitions $\sP^\MPX = \{\P^\MPX_{\alpha^k}\}_{k \geq 0}$ is a hierarchy, and the guarantees of \cref{prp:mpx} still hold. That is, each set in $\P^\MPX_t$ has diameter at most $t$ w.h.p.\ and any edge of weight $w$ crosses $\P^\MPX_t$ with probability at most $O\left(\frac{w \log n}{t}\right)$.
    \end{claim}
    \begin{proof}
        First, note that $\sP^\MPX$ is a hierarchy.
        Take a set $A$ in the partition $\P^\MPX_{\alpha^k}$.
        By definition, any two vertices in $A$ are in the same set of $\P^0_{\alpha^i}$ for $i \geq k$.
        Observe that this immediately implies the same guarantee for $i \geq k + 1$.
        Therefore, all the vertices of $A$ must be in the same set of $\P^\MPX_{\alpha^{k+1}}$, i.e.\ there is a set $B \in \P^\MPX_{\alpha^{k+1}}$ such that $A \subseteq B$.

        The diameter of any set in $\P^\MPX_t$ is at most $t$ with high probability.
        Partition $\P^\MPX_t$ refines $\P^{(0)}_t$.
        Therefore, any set $A \in \P^\MPX_t$ is the subset of some set $B \in \P^0_t$.
        By \cref{prp:mpx}, the diameter of $B$ is at most $t$ with high probability.
        Therefore, since $G$ is a complete graph that satisfies the triangle inequality,
        the diameter of $A \subseteq B$ is also at most $t$.

        The crossing probability of an edge is at most doubled for $\sP^\MPX_t$ due to taking intersections.
        Take an edge $(u, v)$.
        The probability that it crosses $\P^0_{\alpha^i}$ is at most $O\left(\frac{(\log n) w(u, v)}{\alpha^i}\right)$, by \cref{prp:mpx}.
        Also, $(u, v)$ crosses $\P^\MPX_{\alpha^k}$ 
        only if it crosses some $\P^0_{\alpha^i}$ for $i \geq k$.
        Therefore, by taking the union bound over $i$,
        the probability that $(u, v)$ crosses $\P^\MPX_{\alpha^k}$ is at most:
        \begin{align*}
        \sum_{i \geq k} O\left(\frac{w(u, v) \log n}{\alpha^i}\right)
        &\leq O\left(\frac{w(u, v) \log n}{\alpha^k}\right) \cdot \sum_{i \geq 0} \frac{1}{\alpha^i} \\
        &\leq O\left(\frac{w(u, v) \log n}{\alpha^k}\right) \cdot \frac{1}{1 - \frac{1}{\alpha}} \\
        &\leq O\left(\frac{w(u, v) \log n}{\alpha^k}\right),
        \end{align*}
        where the last inequality follows form $\frac{1}{1 - \frac{1}{\alpha}} \leq 2$.
    \end{proof}

    Now, we prove some properties of the hierarchy $\sP$. Recall the definition.
    For each level $t = \alpha^k$, $\P_t$ is obtained by connecting the sets in $\P^\MPX_t$ using edges of weight at most $t$ that do not cross $\P^\MPX_{\alpha t}$.
    That is, $\P_t$ is the set of connected components in $G_{\leq t}[\P^\MPX_{\alpha t}]$.
    
    \begin{claim} \label{clm:P-hierarchy}
        The sequence of decompositions $\sP = \{\P_{\alpha^k}\}_{k \geq 0}$ is a hierarchy, and for all levels $t = \alpha^k$ it holds that
        \begin{enumerate}
            \item any set in $\P_t$ is connected in $G_{\leq t}$,
            \item any set in $\P_t$ has diameter at most $\alpha t$, and
            \item any edge of weight $w \leq t$ crosses $\P_t$ with probability at most $O\left(\frac{w \log n}{\alpha t}\right)$.
        \end{enumerate}
    \end{claim}    

    \begin{proof}
        Observe that $G_{\leq t}[\P^\MPX_{\alpha t}]$ is a subgraph of $G_{\leq \alpha t}[\P^\MPX_{\alpha^2 t}]$.
        Therefore, $\P_t$ is a refinement of $\P_{\alpha t}$, and $\sP$ is a hierarchy.
        Take a set $A \in \P_t$. By definition $A$ is a connected component of $G_{\leq t}[\P^\MPX_{\alpha t}]$.
        Because the edges of $G_{\leq t}[\P^\MPX_{\alpha t}]$ do not cross $\P^\MPX_{\alpha t}$, $A$ is fully contained in a set $B \in \P^\MPX_{\alpha t}.$
        By \cref{clm:Pmpx-hierarchy}, $B$ has diameter at most $\alpha t$.
        Therefore $A$ also has diameter at most $\alpha t$.
        Finally, for an edge with weight $w \leq t$ to cross $\P_t$,
        it must also cross $\P^\MPX_{\alpha t}$,
        which happens with probability at most $O\left(\frac{w \log n}{\alpha t}\right)$ (\cref{clm:Pmpx-hierarchy}).
    \end{proof}

\begin{claim} \label{clm:P-apx-size}
Let $\P_t$ be the set of connected components in $G_{\leq t}[\P^\MPX_{\alpha t}]$. Then, it holds that
$$
\E\left[\sum_t t \left(\card{\P_{t/\alpha}} - \card{\P_t}\right)\right] \leq 2\alpha^2 \cdot \MST(G).
$$
\end{claim}
\begin{proof}
    For any $t = \alpha^k$, let $\C_t$ be the connected components of $G_{\leq t}$.
    Observe that the MST of $G$,
    includes exactly $\card{\C_{t/\alpha}} - \card{\C_t}$ edges of weight $w \in (t/\alpha, t]$.
    If we charge each of these edges $t$ units,
    the total overall charge will be $\alpha \cdot w(\MST(G))$ because each edge is being charged at most $\alpha$ times its weight. Hence, we have
    \begin{equation}
    \sum_t t\left(\card{\C_{t/\alpha}} - \card{\C_t} \right)\leq \alpha \cdot w(\MST(G)). \label{eq:alpha-mst}
    \end{equation}

    To prove the claim, we relate its left-hand side to \eqref{eq:alpha-mst} as follows:

    \begin{equation}
    \sum_t t \left(\card{\P_{t/\alpha}} - \card{\P_t}\right) \leq \sum_t t \left(\card{\P_{t/\alpha}} - \card{\C_{t/\alpha}}\right) + \sum_t t \left(\card{\C_{t/\alpha}} - \card{\C_{t}}\right). \label{eq:alpha-mst-expand}
    \end{equation}
    This is true because on the right-hand side, the terms involving $\C_{t/\alpha}$ cancel out, and $\card{\C_t} \leq \card{\P_t}$ since $\P_t$ refines $\C_t$.

    The second sum on the right-hand side of \eqref{eq:alpha-mst-expand} has already been bounded by $\alpha \cdot w(\MST(G))$ in \eqref{eq:alpha-mst}. It remains to bound the first part.
    Let $\Et_t$ be the edges of $\MST$ with weight $w \in (t/\alpha, t]$, i.e.\ $\C_t = \C_{t/\alpha}\oplus\Et_t$.
    The difference in the number of components
    $\card{\P_{t/\alpha}}-\card{\C_{t/\alpha}}$
    is equal to the number of edges of weight at most $t$ in the MST (i.e.\ the edges in $\bigcup_{t'\leq t/\alpha}\Et_{t'}$) that cross $\P_{t/\alpha}$.
    For an edge in $\Et_{t'}$ the probability of crossing $\P_{t/\alpha}$ is at most $O\left(\frac{t' \log n}{t}\right)$ (by \cref{clm:P-hierarchy}).
    Therefore, we have:
    $$
    \E\left[t \left(\card{\P_{t/\alpha}} - \card{\C_{t/\alpha}}\right) \right]
    \leq 
    t \sum_{t'\leq t/\alpha} \card{\Et_{t'}} \cdot O\left(\frac{t' \log n}{t}\right)
    \leq \sum_{t'\leq t/\alpha} \card{\Et_{t'}} \cdot O\left({t' \log n}\right).
    $$
    Let $L$ be the number of levels. Then, summing over $t$ gives
    \begin{align*}
    \E\left[\sum_t t \left(\card{\P_{t/\alpha}} - \card{\C_{t/\alpha}}\right) \right]
    &\leq
    \sum_{t'} \card{\Et_{t'}} \cdot O\left({t' \log n}\right) \cdot L  \\
    &\leq O\left({L \log n}\right) \cdot \sum_{t'} t' \card{\Et_{t'}} \\
    &\leq O\left({L \log n}\right) \cdot \sum_{t'} t' (\card{\C_{t'/\alpha}}-\card{\C_{t'}}) \\
    &\leq \alpha^2 w(\MST(G)),
    \end{align*}
    where the first inequality holds because each term involving $\card{\Et_{t'}}$
    appears at most $L$ times in the sum, once for each $t \geq t'$.
    The last inequality follows from the choice of $\alpha$ and \eqref{eq:alpha-mst}.

    Combining this with \eqref{eq:alpha-mst} and \eqref{eq:alpha-mst-expand}, we get
    \begin{equation*}
    \E\left[\sum_t t \left(\card{\P_{t/\alpha}} - \card{\P_t}\right)\right] \leq \alpha^2 w(\MST(G)) + \alpha w(\MST(G)) = 2\alpha^2 w(\MST(G)). \qedhere
    \end{equation*}
\end{proof}

\begin{claim} \label{clm:P-apx-mst}
    Let $\P_t$ be the set of connected components in $G_{\leq t}[\P^\MPX_{\alpha t}]$.
    Then it holds that:
    $$
    \E[w(\MST(G_{\leq t} \mid \P_t))] \leq (1 + \epsilon) w(\MST(G)).
    $$
\end{claim}
\begin{proof}
    We construct a tree $T$ that respects $(G_{\leq t} \mid \P_t)$, such that $\E[w(T)] \leq (1 + \epsilon)\MST(G)$.
    Level by level, we add a set of edges $E_t \subseteq G_{\leq t}$ to $T$ such that $\P_t = \P_{t/\alpha} \oplus E_t$, while comparing the weights of the added edges to $w(\MST(G))$ using a charging argument.
    For the charging argument, we define a set of auxiliary weights $\wt$ over the edges of $\MST(G)$. Throughout the process, these weights are updated,
    and every edge $e \in \MST(G)$ charges its auxiliary weight to an edge $e' \in T$, such that $w(e') \leq \wt(e)$.
    Therefore, in the end, it holds that $w(T) \leq \wt(\MST(G))$,
    and we bound $\wt(\MST(G))$.
    
    Let $\Tt = \bigcup_t \Et_t$ be the MST of $G$, such that $\C_t = \C_{t/\alpha} \oplus \Et_t$ (recall $\C_t$ is the set of connected components of $G_{\leq t}$). That is, $\Et_t$ is the set of edges in $\Tt$ with weight $w \in (t/\alpha, t]$. Let $\Ft_t$ be the edges of $\Tt$ with weight at most $t$.
    We start with $T = \emptyset$ and go through the levels $1, \alpha, \alpha^2, \ldots, W$ sequentially.
    Initially, $\wt(e) = w(e)$ for all $e \in \Tt$.
    At any point, we say an edge $e \in \Tt$ is unaccounted for if it has not yet charged anything to any edge in $T$.
    At first, all the edges of $\Tt$ are unaccounted for.
    Each edge will be accounted for on the first level that it is captured.
    
    At level $t$, we \enquote{lift} all the edges of $\Ft_t$ that have not been captured by $\P_t$ to the minimum weight of the next level, meaning we let $\wt(e) = t$ for such an edge $e$. These edges remain unaccounted for.
    The rest of the edges in $\Ft_t$ will be accounted for by the time we are through with this level.
    We go over the unaccounted edges of $\Ft_t$ that have been captured by $\P_t$, and add an edge to $T$ if doing so does not create a cycle.
    In this case, we charge $\wt(e)$ to $e$ itself.
    Note that $\wt(e) \geq w(e)$ and $e$ is accounted for.
    After going through these edges,
    if a set in $\P_t$ is not fully connected by $T$,
    we make it connected by adding edges of weight at most $t$, call these \emph{extra edges}.
    This can be done since $\P_t$ is connected in $G_{\leq t}$ (by \cref{clm:P-hierarchy}).
    We let $E_t$ be the set of added edges for this level.

    Now, we make the charges for the edges of $\Ft_t$ that were captured by $\Ph_t$ but have not yet been accounted for. 
    These are the edges that were lifted from previous levels and captured by this level, yet they could not be added to $T$ because it would have created a cycle.
    Intuitively, an edge $e \in \Tt$ cannot be added because an edge $e' \notin \Tt$ was added in its stead in one of the earlier levels. We aim to charge $\wt(e)$ to this $e'$, which can be found in the cycle $e$ forms with $T$.

    More formally, for the charging argument, we also maintain an auxiliary forest $T'$, which changes with $T$.
    Throughout the process, $T'$ will have the same set of connected components as $T$.
    Whenever an edge is added to $T$, 
    we also add it to $T'$.
    At a level $t$, we modify $T'$ right before the extra edges are added.
    For each $e \in \Ft_t$ that is captured but not accounted for,
    we make the following charge.
    Since $e$ could not be added to $T$, it must form a cycle with $T$.
    Therefore, because $T'$ has the same set of connected components as $T$,
    $e$ also forms a cycle with $T'$.
    At least one edge $e'$ in this cycle is not a member of $\Tt$, because $\Tt$ is a tree.
    We charge $\wt(e)$ to $e'$, and let $T' \gets T' \setminus \{e'\} \cup \{e\}$.
    Note that $w(e') \leq t$ since it must have been added in the previous levels, and $\wt(e)\geq t \geq w(e')$.
    Furthermore, no other edges charge to $e'$ since it is removed from $T'$.
    
    To conclude the proof, note that any edge $e \in \Tt$ is at some point captured and hence accounted for, because the last level of the hierarchy groups all the vertices together.
    Also, when a vertex $e \in \Tt$ charges $\wt(e)$ to an edge $e' \in T$,
    it holds that $\wt(e) \geq w(e')$.
    Therefore, it holds that $w(T) \leq \wt(\Tt)$.
    It remains to analyze $\E[\wt(\Tt)]$.
    For an edge $e \in \Tt$ to be lifted at level $t$, it must have $w(e) \leq t$ and it must cross $\P_t$ which happens with probability $O\left(\frac{w(e) \log n}{\alpha t}\right)$.
    In case it is lifted, $\wt(e)$ is increased by at most $t - t / \alpha \leq t$.
    Therefore, it holds that
    $$
    \E[\wt(e)] \leq w(e) + \sum_t O\left(\frac{w(e) \log n}{\alpha t}\right) \cdot t
    \leq w(e)\left(1 + O\left(\frac{L \log n}{\alpha}\right)\right)
    \leq (1 + \epsilon)w(e),
    $$
    where recall $L$ is the number of levels.
    Summing over $e$, we get
    $$
    \E[\wt(\Tt)] \leq (1 + \epsilon)w(\Tt).
    $$
    Combining this with $w(T) \leq \wt(\Tt)$ gives the result:
    \begin{equation*}
    \E[w(\MST(G_{\leq t} \mid \P_t))]
    \leq \E[w(T)]
    \leq \E[\wt(\Tt)]
    \leq (1 + \epsilon) w(\MST(G)).
    \qedhere
    \end{equation*}
\end{proof}

Now, we move on to analyzing the hierarchy $\sPh$. We prove two claims, similar to those for $\sP$. 

\begin{claim} \label{clm:Ph-apx-size}
    Let the hierarchy $\sPh$ be obtained as in \Cref{alg:main}. Then, it holds that
    $$
    \E\left[\sum_t t \left(\card{\Ph_{t/\alpha}} - \card{\Ph_t}\right)\right] \leq  3\alpha^2 \cdot \MST(G).
    $$
\end{claim}
\begin{proof}
    Take a level $t = \alpha^k$.
    Let us compare $\Ph_t$ and $\P_t$ by examining step \ref{step:mst-leader-compression}.
    The leader compression starts with components $\P^\MPX_t$,
    call this $\P^{(0)}_t$.
    Then, after $r$ rounds of leader compression, a new set of components $\P^{(r)}_t$ is created.
    Finally, all the incomplete components of $\P^{(r)}_t$
    that are in the same set of $\P^\MPX_{\alpha t}$ are joined together to build $\Ph_t$.
    Observe, that any complete component after the leader compression is included in $\Ph_t$ unchanged. That is, a complete component is not merged with any other components.
    Also, since it is complete and does not have any outgoing edges of $G_{\leq t}[\P^\MPX_{\alpha t}]$,
    it also appears in $\P_t$ as is.
    All the incomplete components inside a set of $\P_{\alpha t}$
    are joined together by the algorithm.
    Whereas, in $\P_t$, only some of them are joined together,
    the ones that can be connected using the edges of $G_{\leq t}[\P^\MPX_{\alpha t}]$.
    Therefore, $\P_t$ is a refinement of $\Ph_t$, and we get:
    \begin{equation}
    \sum_t t \left(\card{\Ph_{t/\alpha}} - \card{\Ph_t}\right) \leq
    \sum_t t \left(\card{\P_{t}} - \card{\Ph_t}\right) +
    \sum_t t \left(\card{\P_{t/\alpha}} - \card{\P_t}\right),
    \label{eq:Ph-size}
    \end{equation}
    where the inequality holds because the terms involving $\P_t$ cancel out, and $\card{\Ph_{t/\alpha}}\leq\card{\P_{t/\alpha}}$. The second sum has already been bounded in \cref{clm:P-apx-size}.
    It remains to analyze the first.

    The difference in the number of components of $\P_t$ and $\Ph_t$ is only due to the arbitrary joining of the incomplete components at the end of step \ref{step:mst-leader-compression} since the partitions are the same otherwise.
    More precisely, $\card{\P_t} - \card{\Ph_t}$ is at most equal to the number of sets $A \in \P_t$ that are not completely discovered by leader compression. 
    We bound this number.
    
    Take a set $A \in \P_{t}$.
    Let $N_A^{(r)}$ be the number of sets of $\P^{(r)}_t$ inside $A$.
    We aim to bound the probability of $N_A^{(r)} \neq 1$, i.e.\ $A$ is not completely discovered.
    For any $r$ it holds that $$\E\left[N_A^{(r)} - 1\right]
    \leq \left(\frac{3}{4}\right)^r \left(N_A^{(0)} - 1\right).$$
    To see why, fix $N_A^{(r')}$ for any round $r'$.
    If $N_A^{(r')}$ is equal to $1$, then so is $N_A^{(r'+1)}$.
    Otherwise, during the next round, every component has at least one neighboring component that it can join, which happens with probability at least $\frac{1}{4}$.
    Therefore, conditioned on $N_A^{(r')} \neq 1$, the expectation of $N_A^{(r'+1)}$ is at most $\frac{3}{4}N_A^{(r')}$.
    Combining these two cases together, it always holds that
    $$
    \E\left[N_A^{(r' + 1)} - 1\right] \leq \frac{3}{4}\left(N_A^{(r')} - 1\right).
    $$
    Applying this inductively, gives $\E\left[N_A^{(r)} - 1\right]
    \leq \left(\frac{3}{4}\right)^r \left( N_A^{(0)} - 1\right)$.
    Also, note that since $N^{(r)}_A$ is a positive integral random variable, we have
    $$\Pr\left(N^{(r)}_A \neq 1\right) = \Pr\left(N^{(r)}_A > 1\right) \leq \E\left[N_A^{(r)} - 1\right].$$

    Putting all of this together, we get
    \begin{align}
        \E[\card{\P_t} - \card{\Ph_t}]
        &\leq \sum_{A \in \P_{t}} \Pr\left(N^{(r)}_A > 1\right) \nonumber\\
        &\leq \sum_{A \in \P_{t}} \E\left[N_A^{(r)} - 1\right] \nonumber\\
        &\leq \sum_{A \in \P_{t}} \left(\frac{3}{4}\right)^r\E\left[N_A^{(0)} - 1\right] \nonumber\\
        &= \left(\frac{3}{4}\right)^r \E\left[\card{\P^\MPX_t} - \card{\P_t}\right] \label{eq:1-1}\\
        &\leq \left(\frac{3}{4}\right)^r \E\left[\card{\P_{t/\alpha}} - \card{\P_t}\right] \label{eq:1-2}.
    \end{align}
    Here, \eqref{eq:1-1} follows from the fact that $N^{(0)}_A$ is the number of connected components of $\P^\MPX_t$ inside $A$,
    and \eqref{eq:1-2} holds because $\P_{t/\alpha}$ is a refinement of $\P^\MPX_t$. 
    
    Finally, we get
    \begin{equation}
    \E\left[\sum_t t \left(\card{\P_{t}} - \card{\Ph_t}\right)\right] 
    \leq \left(\frac{3}{4}\right)^r \E\left[\sum_t t \left(\card{\P_{t/\alpha}} - \card{\P_t}\right)\right] \leq \left(\frac{3}{4}\right)^r 2\alpha^2 \cdot w(\MST(G)), \label{eq:p-vs-ph}
    \end{equation}
    where the second inequality was proven in \cref{clm:P-apx-size}.
    Choosing a large enough constant for $r = \Theta\left(\log \frac{1}{\epsilon} + \log \log n\right)$, and plugging this back into \eqref{eq:Ph-size},
    gives the claim.
\end{proof}

\begin{claim} \label{clm:Ph-apx-mst}
    Let $\sPh$ be obtained as in \Cref{alg:main}.
    Then it holds that:
    $$
    \E[w(\MST(G_{\leq \alpha t} \mid \Ph_t))] \leq (1 + \epsilon) w(\MST(G)).
    $$
\end{claim}
\begin{proof}
    We construct a tree $T$ that respects $(G_{\leq \alpha t} \mid \Ph_t)$ as follows.
    Let $\Tt$ be the MST with respect to $(G_{\leq t} \mid \P_t)$,
    and let $\Tt = \bigcup_{t = \alpha^k} \Et_t$ such that $\P_t = \P_{t / \alpha} \oplus \Et_t$.
    Going through the levels one by one, for a level $t = \alpha^k$,
    we add an edge in $\Et_t$ to $T$ if it does not create any cycles.
    After processing the edges of $\Et_t$, some components in $\Ph_t$ might remain unconnected (recall the edges of $\Et_t$ only connect the sets of $\P_t$ which is a refinement of $\Ph_t$).
    To connect them, we use arbitrary edges. 
    The number of these edges is exactly equal to $\card{\P_t} - \card{\Ph_t}$, and they all have weight at most $\alpha t$ since $\Ph_t$ has diameter at most $\alpha t$ (since it is a refinement of $\P^\MPX_{\alpha t})$.

    To analyze the weight of $T$, take the two kinds of edges we have added to it.
    The edges from $\Tt$ that have total weight at most $w(\Tt)$,
    and the arbitrary edges used to connect the remaining components which have total weight at most $\sum_t \alpha t(\card{\P_t} - \card{\Ph_t})$. Therefore, we have
    \begin{align*}
        \E[w(\MST(G_{\leq \alpha t} \mid \Ph_t))]
        &\leq \E[w(T)] \\
        &\leq \E[w(\Tt)] + \E\left[\sum_t \alpha t(\card{\P_t} - \card{\Ph_t})\right] \\
        & \leq (1 + \epsilon)w(\MST(G)) + 
            \left(\frac{3}{4}\right)^r O(\alpha^3) w(\MST(G)) \tag{\cref{clm:P-apx-mst}, and inequality \eqref{eq:p-vs-ph}}\\
        &\leq (1 + O(\epsilon))w(\MST(G)). \tag*{\qedhere}  
    \end{align*}
\end{proof}

For a high-level proof of why $T = \cup_{k\geq 0} E_{\alpha_k}$ approximates the MST w.r.t.\ $(G_{\leq \alpha t} \mid \Ph_t$),
observe that if the modified \Boruvka's algorithm were executed till completion instead of stopping after $r = \Theta\left(\log\frac{1}{\epsilon} + \log \log n\right)$ rounds, then the exact MST w.r.t.\ $(G_{\leq \alpha t} \mid \Ph_t$) would have been recovered.
However, running the subroutine for $r$ rounds guarantees that at most a $\poly\left(\frac{\epsilon}{\log n}\right)$-fraction of the edges of the exact MST are \enquote{missed} because of the early termination. For the arbitrary edges added in their place, we incur a multiplicative error of $\poly(\alpha)$ which is dominated by $2^{\Omega(r)}\epsilon$. 
Therefore, overall we get a $(1 + \epsilon)$-approximation of the MST. Below is the formal proof.

\begin{proof}[Proof of \cref{lem:approx}]
    Let $\Th = \MST(G_{\leq \alpha t}, \mid \Ph_t)$. Recall that this is unique because of the tie-breaking based on IDs.
    Let $\Th = \bigcup_k \Eh_{\alpha^k}$ such that $\Ph_t = \Ph_{t/\alpha} \oplus \Eh_t$, for all $t = \alpha^k$. To prove the claim, we fix a level $t = \alpha^k$, and compare the weight $\Eh_t$ with the weight of $E_t$.
    
    Let $H^{(t)}$ be obtained from $G$ by contracting the vertices in the same set of $\Ph_{t/\alpha}$ (this is what the modified \Boruvka's algorithm starts with).
    Fix a set of vertices $A \in \Ph_{ t}$.
    Since $\sPh$ is a hierarchy, $A$ also corresponds to a set of vertices in $H^{(t)}$.
    Let $E_t[A]$ (resp.\ $\Eh_t[A]$) be the edges of $E_t$ (resp.\ $\Eh_t$) that are inside $A$.
    Similarly, let $H[A]$ be the induced subgraph of $H$ including the vertices in $A$.

    The edges of $E_t^1$ found in step \ref{step:mst-boruvka} of \Cref{alg:main}
    by the modified \Boruvka's algorithm are all in $\Eh_t$.
    Take any set of vertices $A \in \Ph_{ t}$.
    Both $E_t[A]$ and $\Eh_t[A]$ are spanning trees of $H[A]$.
    Observe that for an edge to be added by the modified \Boruvka's algorithm,
    it must have at some point been the minimum edge $e_C$ between a component $C$ and $H[A] \setminus C$.
    Assume for the sake of contradiction that $e_C$
    is not in $\Eh_t$.
    In that case, we can replace an edge in $\Eh_t$ with $e_C$ to create a spanning tree respecting $(G_{\leq \alpha t} \mid \Ph_t)$ that is smaller than $\Th$, which is a contradiction.
    More precisely,
    as $\Eh_t$ is a spanning tree on $H[A]$,
    it contains a path $P$ from one endpoint of $e_C$ to the other.
    Therefore, $P$ must cross between $C$ and $H[A] \setminus C$ using some edge $e'$.
    Then, $\Th \setminus \{e_C\} \cup \{e'\}$
    is a smaller spanning tree w.r.t.\ $(G_{\leq \alpha t} \mid \Ph_t)$,
    which contradicts the fact $\Th$ is the MST.

    Now, we analyze the edges $E_t^2$ found in step \ref{step:mst-join-arbitrarily}, i.e.\ the edges that join the remaining components inside a set of $\Ph_t$ that were not joined by step \ref{step:mst-boruvka}.
    Let $N^{(r')}_A$ be the number of surviving components after $r'$ rounds of the modified \Boruvka{} algorithm, where $0 \leq r' \leq r$.
    We say a component has survived a round of the modified \Boruvka{} algorithm if it does not join any other component.
    $N^{(0)}_A$ is simply the number of vertices in $H[A]$,
    and $N^{(r)}_A$ components survive in $A$ in the end.
    Using the same analysis as the one for the number of surviving connected components of leader compression (see the proof of \cref{clm:Ph-apx-size}),
    we can show that
    $$
    \E[N^{(r)}_A - 1] \leq 2^{-\Omega(r)} \big(\card{V(H[A])} - 1\big).
    $$

    The number of edges of $E_t^2$ added inside $A$,
    is $N^{(r)}_A - 1$,
    and each of them has weight at most $\alpha t$,
    since $\Ph_t$ has diameter at most $\alpha t$.
    Therefore, we can bound the total weight of the edges in $E_t^2$ as follows
    \begin{align}
        E[w(E_t^2)] &\leq
        \sum_{A \in \Ph_t} 2^{-\Omega(r)}(\card{V(H[A])} - 1) \cdot \alpha t \nonumber \\
        &=  2^{-\Omega(r)}\alpha t \sum_{A \in \Ph_t} (\card{V(H[A])} - 1) \nonumber \\
        &=  2^{-\Omega(r)}\alpha t \cdot \left(\card{\Ph_{t/\alpha}} - \card{\Ph_t}\right), \nonumber
    \end{align}
    where the last equality follows from $\sum_{A \in |\Ph_t|} 1 = |\Ph_t|$, and $\sum_{A \in |\Ph_t|} |V(H[A])| = |\Ph_{t/\alpha}|$.
    
    Combining this with the fact that the edges of $E_t^1$ have total weight at most $w(\Eh_t)$ and summing over $t$, we get
    \begin{align*}
    w(T) &= \sum_{t} w(E_t^1) + w(E_t^2) \\
    &\leq w(\Th) +  2^{-\Omega(r)} \sum_{t} \alpha t \cdot \left(\card{\Ph_{t/\alpha}} - \card{\Ph_t}\right)\\
    &\leq w(\Th) + 2^{-\Omega(r)} O(\alpha^3) w(\MST(G)) \\
    &\leq (1 + O(\epsilon)) \MST(G),
    \end{align*}
    where the next to last inequality is by \cref{clm:Ph-apx-size},
    and the last inequality by \cref{clm:Ph-apx-mst} and the choice of $r$.
    This concludes the proof.
\end{proof}

\subsection{Runtime Analysis and Implementation Details}

We prove the following:

\begin{lemma} \label{lem:implementation}
    \Cref{alg:main} can be implemented in $O\left(\log \frac{1}{\epsilon} + \log \log n\right)$ rounds of the MPC model, with $O(n^\delta)$ space per machine and $\Ot(n^2)$ total space, where $0 < \delta < 1$.
\end{lemma}

We assume that the metric is given to us as a $n \times n$ matrix, where the minimum entry is at least $1$, the maximum entry is at most $W = \poly(n)$, and any weight fits in one word. We use two tools repeatedly in our implementation:

\begin{enumerate}
    \item \textbf{Sorting \cite{goodrich2011sorting}:}
    A sequence of $N$ elements can be sorted in $O\left(\frac{1}{\delta}\right)$ rounds using $O(N^\delta)$ space per machine and $O(N)$ total space.
    \item \textbf{\boldmath$\left(\frac{1}{\delta}\right)$-ary Trees:}
    A message of size $\Ot(1)$ that is stored in one machine can be communicated to $N^c$ machines in $O\left(\frac{1}{\delta}\right)$ rounds, where each machine has local memory $O(N^\delta)$.
\end{enumerate}
The latter can be used to compute simple operations on a large number of elements. For example, a mapping or filtering function $f$ that can be conveyed using $\Ot(1)$ words, can be computed on $O(N)$ elements using $O(N^{1 - \delta})$ machines and $O\left(\frac{1}{\delta}\right)$ as follows. First, the elements are stored in $O(N^{1-\delta})$ machines. Then, $f$ is communicated to them using a $\left(\frac{1}{\delta}\right)$-ary tree.

\begin{claim}
    The hierarchy $\sP^\MPX$ can be constructed in $O\left(\frac{1}{\delta}\right)$ rounds in the MPC model.
\end{claim}
\begin{proof}
    First, we need to construct $\P^0$.
    The data is replicated $L$ times, one copy per level.
    Then, for each level $t = \alpha^k$, the partition $\P^0_t$ is built independently.
    A set of $O(n^{1-\delta})$ machines $M_u$ are created for each vertex $u$.
    Using sorting the edges $(u, v)$ for all other $v$ are transported to $M_u$.
    Then, a predesignated machine in $M_u$ draws the delay $\delta_u$ and communicates it to the rest of $M_u$.
    At this point, $M_u$ creates a tuple $\big(v,\ u,\ w(u, v) - \delta_u\big)$
    for each vertex $v$.
    This tuple represents $(\textnormal{vertex},\ \textnormal{potential center},\  \textnormal{shifted distance})$.
    Finally, using sorting, $M_u$ collects all the tuples $\big(u,\ v,\ w(u, v) - \delta_v\big)$ and identifies the vertex $v$ with the smallest shifted distance as its center.
    The component IDs in $\P^0_t$ are simply the vertex ID of the center for that component.

    To create $\sP^\MPX$, we need to intersect the partitions in $\sP^0$.
    Recall two vertices are in the same set of $\P^\MPX_{\alpha^k}$ if they are in the same set of $\P^0_{\alpha^i}$ for all $i \geq k$. That is, the component ID of a vertex $u$ in $\P^\MPX_{\alpha^k}$ can be thought of as a tuple $\P^\MPX(u) = \left(\P^0_{\alpha^k}(u),\ \P^0_{\alpha^{k + 1}}(u),\ \ldots,\ \P^0_W(u)\right)$,
    where $\P^0_t(u)$ is the component ID of $u$ in $\P^0_t$ that was computed in the previous step.

    Therefore, to compute $\P^\MPX$, we can simply create one machine $m_u$ for each vertex $u$.
    Then, $m_u$ collects the component IDs $\P^0_t(u)$ for all $t = \alpha^k$,
    and creates all the component IDs $\P^\MPX_t(u)$ as described above.
    After these IDs are created, we can use sorting to index the tuples on each level, and use one word for each component ID of $\P^\MPX_t$ rather than a tuple.
    For the purposes of the next steps of the algorithm, for each vertex $u$ and level $t$, instead of simply storing its component ID in $\P^\MPX_t$, we store the two-word tuple $\big(\P^\MPX_t(u),\ \P^\MPX_{\alpha t}(u)\big)$.
\end{proof}

\begin{claim}
    The hierarchy $\sPh$ can be computed in $O\left(\frac{1}{\delta} \left(\log \frac{1}{\epsilon} + \log \log n\right)\right)$ rounds of MPC.
\end{claim}
\begin{proof}
    For each level $t$, $\Ph_t$ is computed independently in parallel. A copy of all the weights is obtained for each level, and $\P^\MPX_t$ from the previous step is used.
    
    First, we show that each round of leader compression can be implemented in $O\left(\frac{1}{\delta}\right)$ rounds.
    Let the partition $\P^{(r)}_t$ be the result of $r$ rounds of leader compression.
    $\P^{(0)}_t = \P^\MPX_t$ is computed by the previous step.
    Given $\P^{(r)}_t$, we show how to compute $\P^{(r+1)}_t$.

    We designate a set of $O(n^{1-\delta})$ machines $M_u$ for each vertex $u$,
    and let a predesignated machine $m_u$ in $M_u$ be in charge of the whole group.
    All these machines are given the tuple $\P^\MPX_t(u)$, $\P^\MPX_{\alpha t}(u)$, and $\P^{(r)}_t(u)$.
    Using sorting, we collect all the edges of $u$ in $M_u$. Then, by communicating over the edges,
    for each edge, we attach to it the three component IDs in $\P^\MPX_t$, $\P^\MPX_{\alpha t}$, and $\P^{(r)}_t$ for both endpoints.
    Also, for each component in $C \in \P^{(r)}_t$, we create a group of machines $M_C$.
    These machines form a $O\left(\frac{1}{\delta}\right)$-ary tree where the leaves are the machines $m_u$ in charge of the vertices $u \in C$. 
    Also, one machine $m_C$ in $M_C$ is designated in charge of $C$.

    To begin the round, the machine $m_C$ draws a fair coin $x_C$, for each component $C$.
    Then, $x_C$ is sent down the tree in $M_C$ to reach all the machines $m_u$ for $u \in C$.
    From there, $x_C$ is sent to all the machines in $M_u$.
    Now, each machine in $M_u$ can communicate $x_C$ over its edges, and check whether it holds an edge in $G_{\leq t}[\P^\MPX_{\alpha t}]$ that exits $C$. That is, for an edge $e = (u, v)$, it must check:
    \begin{enumerate}
        \item it has weight at most $t$, i.e.\ $w(u, v) \leq t$,
        \item the endpoints are in different sets of $\P^{(r)}_t$, i.e.\ $\P^{(r)}_t(u) \neq \P^{(r)}_t(v)$,
        \item the edge does not cross $\P^\MPX_{\alpha t}$, i.e.\ $\P^\MPX_{\alpha t}(u) = \P^\MPX_{\alpha t}(v)$,
        \item the coin flips are in order, i.e.\ $x_{\P^{(r)}_t(u)} = 0$, and $x_{\P^{(r)}_t(v)} = 1$.
    \end{enumerate}
    If all these conditions are met, then $C$ can potentially join the component $C'$ on the other endpoint, and $e$ is sent up the tree to $m_C$.

    Possessing an edge $e_C$, the component $C$, joins the component at the other end. To do so, $m_C$ sends the ID down the tree to the machines $m_u$ for all $u \in C$, where $m_u$ creates $\P^{(r+1)}_t(u)$. Afterwards, the trees are dismantled, and the data is sorted to prepare for the next round. This concludes the description of a round of leader compression.

    To do the arbitrary joining after leader compression is finished, we need to identify the incomplete components. To do so, we employ a similar strategy to a round of leader compression. Except, when a component $C$ finds an edge $e_C$ that it can join over, instead of joining the other side, it simply marks itself as incomplete.
    Afterwards,
    for each vertex $u$ in an incomplete component, a tuple $\left(u, \P^{(r)}_t(u), \P^\MPX_{\alpha t}(u)\right)$ is created.
    The tuples are sorted by the last entry, and a tree is created over each group of tuples with the same component $A \in \P^\MPX_{\alpha t}$.
    Then, an arbitrary incomplete component of $\P^{(r)}_t$ inside $A$ is chosen and all the other incomplete components join it. 
\end{proof}

\begin{claim}
    The tree $T$ in \Cref{alg:main} can be computed from $\sPh$ in $O\left(\frac{1}{\delta}\left(\log\frac{1}{\epsilon} + \log\log n\right)\right)$ rounds in the MPC model.
\end{claim}
\begin{proof}
    The implementation is essentially the same as that of leader compression. The main difference is that the valid edges to join over are $G_{\leq \alpha t}[\Ph_t]$, and instead of an arbitrary edge, the one with minimum weight is chosen.
\end{proof}
\section{A Lower Bound for Metric MST} \label{sec:lb}

In this section, we prove \cref{thm:lb}, restated below.
\thmlb*

First, we make some definitions.
In a graph $G$, the $r$-neighborhood of a vertex is the set of vertices that are at most $r$ edges away from it.
The $(1, 2)$-metric obtained from $G$, is defined as a metric where two vertices have distance $1$ if and only if there is an edge between them in $G$.
For the purposes of this section, we use \enquote{distance} to refer to their hop-distance in $G$. To refer to the distances in the metric, we explicitly mention it.

We also use a well-known operation in MPC algorithms, sometimes referred to as doubling:
\begin{definition}
    Given a graph $G$, a round of \emph{doubling} adds all the 2-hop edges to the graph.
    That is, the graph $G'$ obtained from a round of doubling on $G$, includes an edge $(u, v)$ if and only if $u$ and $v$ have distance at most $2$ in $G$.
\end{definition}

Given $\poly(n)$ total space, a round of doubling can be easily implemented in $O(1)$ rounds of the MPC model.

\thmlb*

\begin{proof}
    To prove the theorem, we show that the existence of an algorithm for $(1 + \epsilon)$-approximation of the MST in $(1, 2)$-metrics that runs in $R$ rounds, implies an algorithm that solves the \onetwocycle{} problem in $O(R \cdot \log_{1/\epsilon} n)$ rounds.
    Observe that if $R = o\left(\log \frac{1}{\epsilon}\right)$, this implies an $o(\log n)$ algorithm for \onetwocycle{}, which is a contradiction and concludes the proof.

    Throughout the proof, we assume that $\frac{1}{\epsilon} = O(n^\delta)$, where $O(n^\delta)$ is the local space on each machine and $0 < \delta < 1$.
    This is without loss of generality since any $(1 + \frac{1}{n})$-approximation is a $(1 + \frac{1}{n^\delta})$-approximation, and $\log n = \Theta(\log n^\delta)$.
    
    To make the reduction, we present an algorithm that given a cycle of length $n$,
    with high probability outputs the order in which the vertices appear in the cycle.
    Given any input graph, the output of the algorithm can be used to validate if the input is a cycle of length $n$.
    That is, given the output ordering $u_0, \ldots u_{n-1}$, one can simply check whether there are $n$ edges, one between $u_i$ and $u_{(i+1) \bmod n}$ for each $0\leq i\leq n-1$.

    To construct the ordering of the vertices, we repeatedly use two subroutines: one that for each vertex, discovers its $\frac{1}{\epsilon}$-neighborhood using the approximate MST algorithm, and one subroutine that uses the former to break down the cycle into smaller cycles.
    Executing the latter subroutine $k$ times on the graph, results in $\frac{1}{\epsilon^k}$ cycles of length $\epsilon^k n$. We argue later that without loss of generality, one can assume $\frac{1}{\epsilon}$ is an integer, and $n$ is a power of $\frac{1}{\epsilon}$, i.e.\ $\epsilon^k n$ is an integer.
    \begin{claim}\label{clm:neighborhood-detection}
        Given a graph $G$, consisting of $\frac{1}{\epsilon^k}$ cycles of length $\epsilon^k n$, we can find the $\frac{1}{\epsilon}$-neighborhood of every vertex with high probability, using $O(n^\delta)$ space per machine, $\poly(n)$ total space, and $O(R)$ rounds.
    \end{claim}
    Observe, the above implies that for each vertex $u$, we can find the subpath of length $\frac{2}{\epsilon}$ that is centered at $u$ by gathering all the vertices in its $\frac{1}{\epsilon}$-neighborhood in one machine. A cycle of length $L$ can then be broken down into $\frac{1}{\epsilon}$ cycles of length $\epsilon L$ by removing all the edges, and connecting each vertex to the two vertices that were at distance exactly $\frac{1}{\epsilon}$ (\cref{fig:cycle-dividing}). More formally,
    
    \begin{claim} \label{clm:cycle-breakdown}
        Given a graph $G$, consisting of $\frac{1}{\epsilon^k}$ cycles $C_0, \ldots, C_{1/\epsilon^k-1}$ of length $\epsilon^k n$, and a list of vertices $u^{0}_0, \ldots, u^{1/\epsilon^k-1}_0$, such that the cycle $C_i$ contains vertices $u^{i}_0, \ldots, u^{i}_{\epsilon^k n - 1}$ (in that order),
        using $O(n^\delta)$ space per machine,
        $\poly(n)$ total space,
        and $O(R)$ rounds,
        we can construct a new graph $G'$ on the same vertex set, such that for each $i$,
        there are $\frac{1}{\epsilon}$ cycles $C^{i}_0, \ldots, C^{i}_{1/\epsilon - 1}$ in $G'$, such that $C^{i}_j$ consists of vertices $u^{i}_j, u^{i}_{j + 1/\epsilon}, u^{i}_{j + 2/\epsilon}, \ldots, u^{i}_{j + (\epsilon^{k+1}n - 1)/\epsilon}$ in that order.
        Furthermore, we can compute the list of vertices $u^{i}_j$, where $0 \leq i \leq 1/\epsilon^k -1$ and $0 \leq j \leq 1/\epsilon-1$.
        That is, one vertex $u^i_j$ per cycle $C^i_j$, such that for each $i$ they form a subpath in $C_i$.
    \end{claim}

    \begin{figure}[h]
        \centering
        \begin{tikzpicture}[scale=2.5]

\def\n{15}

\definecolor{lightred}{RGB}{255, 182, 193}
\definecolor{lightblue}{RGB}{173, 216, 230}
\definecolor{lightgreen}{RGB}{144, 238, 144}
\definecolor{darkred}{RGB}{200, 0, 0}
\definecolor{darkblue}{RGB}{0, 0, 139}
\definecolor{darkgreen}{RGB}{0, 100, 0}

\begin{scope}[shift={(-0.5,0)}]
    \foreach \i in {0,1,...,\n} {
        \node[draw, circle, inner sep=1.5pt, fill=white, minimum size=15pt] (a\i) at ({360/\n * \i}:1) {};
    }

    \foreach \i in {1,...,\n} {
        \node at ({360/\n * \i}:1) {\scriptsize $v_{\i}$};
    }

    \foreach \i in {1,...,\n} {
        \pgfmathtruncatemacro{\j}{mod(\i+1,\n)}
        \draw[black, thick] (a\i) to (a\j);
    }

    \node at (0, -1.5) {\textbf{(a)}};
\end{scope}

\draw[dotted] (1, -1) -- (1, 1);

\begin{scope}[shift={(2.5,0)}]
    \foreach \i in {0,...,\n} {
        \node[draw, circle, inner sep=1.5pt, fill=white, minimum size=15pt] (b\i) at ({360/\n * \i}:1) {};
    }

    \foreach \i in {0,...,\n} {
        \node at ({360/\n * \i}:1) {\scriptsize $v_{\i}$};
    }

    \foreach \i in {1,4,...,\n} {
        \pgfmathtruncatemacro{\j}{mod(\i+3,\n)}
        \draw[darkred, thick, bend left=20] (b\i) to (b\j);
        \node[draw, circle, inner sep=1.5pt, fill=lightred, minimum size=15pt] (b\i) at ({360/\n * \i}:1) {};
        \node at ({360/\n * \i}:1) {\scriptsize $v_{\i}$};
    }

    \foreach \i in {2,5,...,\n} {
        \pgfmathtruncatemacro{\j}{mod(\i+3,\n)}
        \draw[darkblue, thick, bend left=20] (b\i) to (b\j);
        \node[draw, circle, inner sep=1.5pt, fill=lightblue, minimum size=15pt] (b\i) at ({360/\n * \i}:1) {};
        \node at ({360/\n * \i}:1) {\scriptsize $v_{\i}$};
    }

    \foreach \i in {3,6,...,\n} {
        \pgfmathtruncatemacro{\j}{mod(\i+3,\n)}
        \draw[darkgreen, thick, bend left=20] (b\i) to (b\j);
        \node[draw, circle, inner sep=1.5pt, fill=lightgreen, minimum size=15pt] (b\i) at ({360/\n * \i}:1) {};
        \node at ({360/\n * \i}:1) {\scriptsize $v_{\i}$};
    }

    \node at (0, -1.5) {\textbf{(b)}};
\end{scope}

\end{tikzpicture}
        \caption{An example of breaking down a cycle of length $L = 15$ into $\frac{1}{\epsilon} = 3$ cycles of length $\epsilon L = 5$. To break down the cycle (a), a new graph (b) is constructed on the same vertex set, where each vertex is connected to the two vertices at distance $\frac{1}{\epsilon}$, e.g.\ $v_2$ is connected to $v_{14}$ and $v_{5}$.
        Given the subpath $v_1$-$v_2$-$v_3$, and the order of the vertices in each cycle of (b), the order of the vertices in (a) can be constructed as follows.
        Start the order with the first vertex of each cycle order, i.e.\ $v_1, v_2, v_3$.
        Then, add the second vertex of each cycle to the order, i.e.\ $v_4, v_5, v_6$.
        Then, the third vertex of each cycle, and so on.}
        \label{fig:cycle-dividing}
    \end{figure}
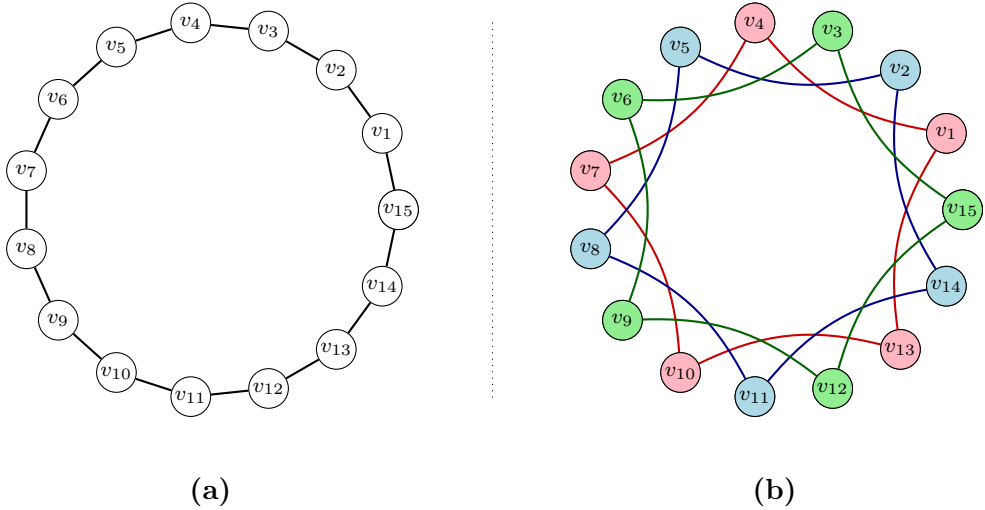

    First, we show how to utilize \cref{clm:cycle-breakdown} to find the order of the cycles.
    Assume for simplicity that $\frac{1}{\epsilon}$ is an integer. This can be done by taking a smaller parameter $\epsilon'$ such that $\frac{\epsilon}{2} \leq \epsilon' \leq \epsilon$.
    Also, we can assume that the number of vertices $n$ is a power of $\frac{1}{\epsilon}$.
    This is without loss of generality
    since we can round up $n$ to the nearest power of $\frac{1}{\epsilon}$, say $n'$, by adding $n' - n$ extra edges anywhere in the cycle.
    As a result, the number of vertices grows by at most a factor of $\frac{1}{\epsilon}$, and is still $\poly(n)$.

    Given a cycle of length $n = \frac{1}{\epsilon^N}$, we can find the order of the vertices by recursively applying \cref{clm:cycle-breakdown}.
    Take any vertex $u_0$, and let $u_0, \ldots, u_{n-1}$ be the order of the cycle.
    To start, we can feed the graph and $u_0$ to \cref{clm:cycle-breakdown}.
    We get a new graph consisting of the cycles $C_0, \ldots, C_{1/\epsilon-1}$ of length $\frac{1}{\epsilon^{N-1}}$. Note that we do not know what the cycles are. Rather, we know of only one vertex in each cycle $C_i$, namely $u_i$, where $u_0,\ldots,u_{1/\epsilon-1}$ form a subpath in the cycle.

    After computing the order of vertices in each cycle $C_i$ recursively, we can recover the order in the original cycle by \enquote{shuffling} the orders of $C_i$'s as follows. First, we add the first vertex in each cycle to the order, i.e.\ $u_0, \ldots, u_{1/\epsilon-1}$.
    Then we add the second vertex of each cycle to the order, these are the $\frac{1}{\epsilon}$ vertices that follow $u_0, \ldots, u_{1/\epsilon-1}$.
    Then, the third, and so on (\Cref{fig:cycle-dividing}).

    The base case is when the cycles have length $\frac{1}{\epsilon}$. We can simply use $O\left(\log \frac{1}{\epsilon}\right)$ rounds of doubling to gather the whole cycle in one machine, and then order the vertices locally.

    Every step of this recursion takes $O(R)$ rounds, recall $R$ is the number of rounds the approximate MST algorithm takes to finish. There are $N = \log_{1/\epsilon}n$
    levels of recursion. In conclusion, we have presented an algorithm that orders the vertices of a cycle of length $n$ in $O(R \cdot \log_{1/\epsilon}n) = O\left(R \cdot \frac{\log n}{\log \frac{1}{\epsilon}}\right)$ rounds.
    Therefore, an $R = o\left(\log \frac{1}{\epsilon}\right)$-round algorithm for MST approximation implies a $o(\log n)$ algorithm for ordering the vertices of a cycle, which in turn implies a $o(\log n)$ algorithm to solve the \onetwocycle{} problem.
    This concludes the proof.
\end{proof}

The remainder of this section is devoted to proving \cref{clm:neighborhood-detection,clm:cycle-breakdown}. We begin by presenting some auxiliary tools and claims used in the proofs.

\begin{definition} \label{def:two-switch}
    Given a graph $G$, doing a \emph{two-switch} on two of its edges $(u_1, v_1)$, $(u_2, v_2)$ is defined as removing them from the graph and inserting $(u_1, u_2)$ and $(v_1, v_2)$.
\end{definition}

\begin{claim} \label{clm:two-switch}
    Let $G'$ be obtained from $G$ by doing a two-switch on $(u_1, v_1)$, $(u_2, v_2)$, where $G$ is a disjoint union of cycles. Then, it holds:
    \begin{enumerate}
        \item If $u_1$ and $u_2$ are in different cycles in $G$, then the two cycles are joined in $G'$. \label{item:two-switch-two-cycles}
        \item If $u_1$ and $u_2$ are in the same cycle in $G$, and $v_1$ and $v_2$ lie on different paths from $u_1$ to $u_2$,
        then the vertices of the cycle still form a cycle in $G'$.
        \label{item:two-switch-one-cycle-different-paths}
        \item If $u_1$ and $u_2$ are in the same cycle in $G$, and $v_1$ and $v_2$ lie on the same path from $u_1$ to $u_2$,
        then the cycle is borken down into two cycles in $G'$, one containing $u_1$ and $u_2$, and another containing $v_1$ and $v_2$.
        \label{item:two-switch-one-cycle-same-paths}
    \end{enumerate}
\end{claim}
\begin{proof}
    Proof by picture: see \Cref{fig:two-switch-two-cycles,fig:two-switch}.
    \begin{figure}
        \begin{tikzpicture}[scale=0.5, transform shape, every node/.style={font=\Large}]
    \definecolor{lightblue}{RGB}{173, 216, 230}
    \definecolor{lightred}{RGB}{255, 182, 193}

    \foreach \i in {1,...,10}
    {
        \ifthenelse{\i=1 \OR \i=2}
            {\node[draw, circle, fill=lightblue, minimum size=1.5em] (v\i) at ({-18 + 360/10 * (\i - 1)}:3) {\Large $v_{\i}$};}
            {\node[draw, circle, minimum size=1.5em] (v\i) at ({-18 + 360/10 * (\i - 1)}:3) {\Large $v_{\i}$};}
    }
    
    \foreach \i in {1,...,9}
    {
        \pgfmathtruncatemacro\next{\i + 1}
        \draw (v\i) -- (v\next);
    }
    \draw (v10) -- (v1);
    
    \begin{scope}[xshift=8cm]
        \foreach \i [evaluate=\i as \j using int(\i+10)] in {1,...,10}
        {
            \ifthenelse{\j=15 \OR \j=16}
                {\node[draw, circle, fill=lightred, minimum size=1.5em] (v\j) at ({18 + 360/10 * (\i - 1)}:3) {\Large $v_{\j}$};}
                {\node[draw, circle, minimum size=1.5em] (v\j) at ({18 + 360/10 * (\i - 1)}:3) {\Large $v_{\j}$};}
        }
        
        \foreach \i in {11,...,19}
        {
            \pgfmathtruncatemacro\next{\i + 1}
            \draw (v\i) -- (v\next);
        }
        \draw (v20) -- (v11);
    \end{scope}
    
    \node at (4, -6) {\Large \textbf{(a)}};
    
    \draw[dotted] (12, -6) -- (12, 6);
    
    \begin{scope}[xshift=16cm]
        \foreach \i in {1,...,10}
        {
            \ifthenelse{\i=1 \OR \i=2}
                {\node[draw, circle, fill=lightblue, minimum size=1.5em] (v\i) at ({-18 + 360/10 * (\i - 1)}:3) {\Large $v_{\i}$};}
                {\node[draw, circle, minimum size=1.5em] (v\i) at ({-18 + 360/10 * (\i - 1)}:3) {\Large $v_{\i}$};}
        }
        
        \foreach \i in {1,...,9}
        {
            \pgfmathtruncatemacro\next{\i + 1}
            \ifthenelse{\i=1 \OR \i=2}
                {}
                {\draw (v\i) -- (v\next);}
        }
        \draw (v10) -- (v1);
        
        \begin{scope}[xshift=8cm]
            \foreach \i [evaluate=\i as \j using int(\i+10)] in {1,...,10}
            {
                \ifthenelse{\j=15 \OR \j=16}
                    {\node[draw, circle, fill=lightred, minimum size=1.5em] (v\j) at ({18 + 360/10 * (\i - 1)}:3) {\Large $v_{\j}$};}
                    {\node[draw, circle, minimum size=1.5em] (v\j) at ({18 + 360/10 * (\i - 1)}:3) {\Large $v_{\j}$};}
            }
            
            \foreach \i in {11,...,19}
            {
                \pgfmathtruncatemacro\next{\i + 1}
                \ifthenelse{\i=15 \OR \i=16}
                    {}
                    {\draw (v\i) -- (v\next);}
            }
            \draw (v20) -- (v11);
        \end{scope}
        
        \draw (v1) -- (v16);
        \draw (v2) -- (v15);
        \draw (v2) -- (v3);
        \draw (v16) -- (v17);
        
        \node at (4, -6) {\Large \textbf{(b)}};
    \end{scope}
\end{tikzpicture}
        \caption{A possible outcome of a two-switch.
        (b) is obtained from (a) by doing a two-switch on $(v_1, v_2)$ and $(v_{16}, v_{15})$.
        Note that $v_1$ and $v_{16}$ are on different cycles here (case \ref{item:two-switch-two-cycles} in \cref{clm:cycle-breakdown}).}
        \label{fig:two-switch-two-cycles}
    \end{figure}
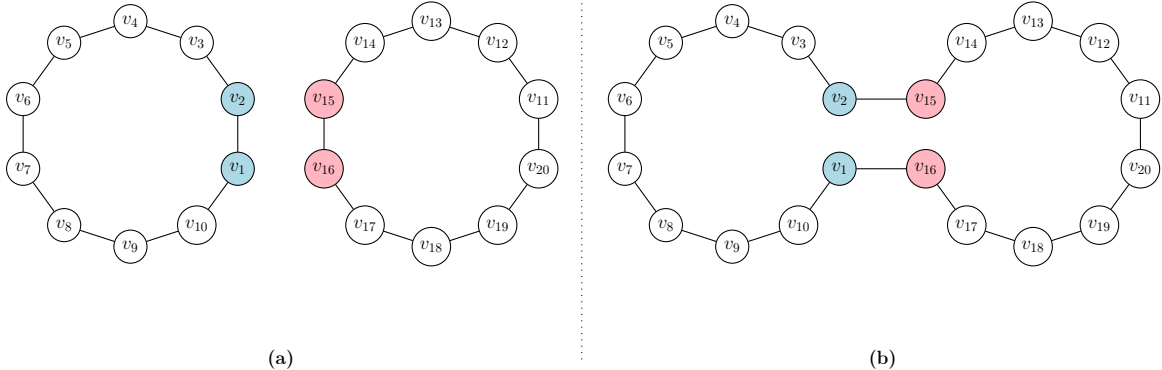
    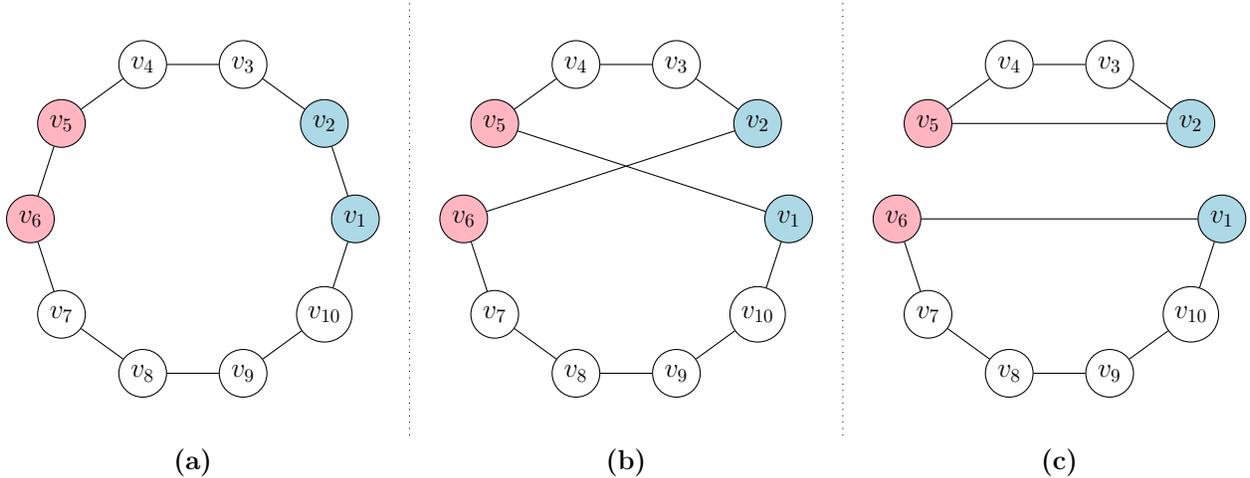
\begin{figure}
        \begin{tikzpicture}[scale=0.72, transform shape, every node/.style={font=\Large}]
    \definecolor{lightblue}{RGB}{173, 216, 230}
    \definecolor{lightred}{RGB}{255, 182, 193}

    \foreach \i in {1,...,10}
    {
        \ifthenelse{\i=1 \OR \i=2}
            {\node[draw, circle, fill=lightblue, minimum size=1.5em] (v\i) at ({360/10 * (\i - 1)}:3) {\Large $v_{\i}$};}
            {\ifthenelse{\i=5 \OR \i=6}
                {\node[draw, circle, fill=lightred, minimum size=1.5em] (v\i) at ({360/10 * (\i - 1)}:3) {\Large $v_{\i}$};}
                {\node[draw, circle, minimum size=1.5em] (v\i) at ({360/10 * (\i - 1)}:3) {\Large $v_{\i}$};}
            }
    }
    
    \foreach \i in {1,...,9}
    {
        \pgfmathtruncatemacro\next{\i + 1}
        \draw (v\i) -- (v\next);
    }
    \draw (v10) -- (v1);
    
    \node at (0, -4.5) {\Large \textbf{(a)}};
    
    \draw[dotted] (4, -4) -- (4, 4);
    
    \begin{scope}[xshift=8cm]
        \foreach \i in {1,...,10}
        {
            \ifthenelse{\i=1 \OR \i=2}
                {\node[draw, circle, fill=lightblue, minimum size=1.5em] (v\i) at ({360/10 * (\i - 1)}:3) {\Large $v_{\i}$};}
                {\ifthenelse{\i=5 \OR \i=6}
                    {\node[draw, circle, fill=lightred, minimum size=1.5em] (v\i) at ({360/10 * (\i - 1)}:3) {\Large $v_{\i}$};}
                    {\node[draw, circle, minimum size=1.5em] (v\i) at ({360/10 * (\i - 1)}:3) {\Large $v_{\i}$};}
                }
        }
        
        \foreach \i in {1,...,9}
        {
            \pgfmathtruncatemacro\next{\i + 1}
            \ifthenelse{\i=1 \OR \i=5}
                {}
                {\draw (v\i) -- (v\next);}
        }
        \draw (v10) -- (v1);
        
        \draw (v1) -- (v5);
        \draw (v2) -- (v6);
        
        \node at (0, -4.5) {\Large \textbf{(b)}};
    \end{scope}
    
    \draw[dotted] (12, -4) -- (12, 4);
    
    \begin{scope}[xshift=16cm]
        \foreach \i in {1,...,10}
        {
            \ifthenelse{\i=1 \OR \i=2}
                {\node[draw, circle, fill=lightblue, minimum size=1.5em] (v\i) at ({360/10 * (\i - 1)}:3) {\Large $v_{\i}$};}
                {\ifthenelse{\i=5 \OR \i=6}
                    {\node[draw, circle, fill=lightred, minimum size=1.5em] (v\i) at ({360/10 * (\i - 1)}:3) {\Large $v_{\i}$};}
                    {\node[draw, circle, minimum size=1.5em] (v\i) at ({360/10 * (\i - 1)}:3) {\Large $v_{\i}$};}
                }
        }
        
        \foreach \i in {1,...,9}
        {
            \pgfmathtruncatemacro\next{\i + 1}
            \ifthenelse{\i=1 \OR \i=5}
                {}
                {\draw (v\i) -- (v\next);}
        }
        \draw (v10) -- (v1);
        
        \draw (v1) -- (v6);
        \draw (v2) -- (v5);
        
        \node at (0, -4.5) {\Large \textbf{(c)}};
    \end{scope}
\end{tikzpicture}
        \caption{Two possible outcomes of a two-switch.
        (b) is obtained from (a) by doing a two-switch on $(v_1, v_2)$ and $(v_5, v_6)$. Note that $v_6$ and $v_2$ are on different paths from $v_1$ to $v_5$ (case \ref{item:two-switch-one-cycle-different-paths} in \cref{clm:cycle-breakdown}).
        (c) is obtained from (a) by doing a two-switch on $(v_1, v_2)$ and $(v_6, v_5)$. Observe that $v_5$ and $v_2$, are on the same path from $v_1$ to $v_6$ (case \ref{item:two-switch-one-cycle-same-paths} in \cref{clm:cycle-breakdown}).}
        \label{fig:two-switch}
    \end{figure}
\end{proof}

Given a graph $G$ with $n$ vertices that is a disjoint union of $c$ cycles, we relate it to MST approximation in $(1, 2)$-metrics as follows.
Let two points in the metric have distance $1$ if there is an edge between them in $G$, and distance $2$ otherwise.
An MST of this metric has weight $n + c - 2$,
as it uses $n - c$ edges of weight $1$, and $c - 1$ edges of weight $2$.
Alternatively, this can be viewed as turning $G$ into a forest by removing $c$ edges, and then connecting the components to build the MST. Similarly, for an approximate MST,
the problem can be viewed as turning $G$ into a forest by removing not many more than $c$ edges, and then connecting the components using edges of weight at most $2$.

\begin{definition} \label{def:removed-edges}
    Given a $(1, 2)$-metric, 
    we say a weight-$1$ edge is \emph{removed} by a spanning tree,
    if it is not included in it.
\end{definition}

\begin{claim} \label{clm:edge-removal-apx-mst}
    Let $G$ be a disjoint union of cycles, and let $T$ be a $(1 + \epsilon)$-approximate MST of the $(1, 2)$-metric obtained from $G$.
    Then, $T$ removes at most $c + \epsilon(n + c - 2)$ edges of $G$, where $n$ is the number of vertices in $G$, and $c$ is the number of cycles.
\end{claim}
\begin{proof}
Observe that any spanning tree that removes $k$ edges of $G$, contains $n - k$ weight-$1$ edges, and $k - 1$ weight-$2$ edges. Therefore, it has weight $n + k - 2$ edges.
As mentioned, the MST has weight $n + c - 2$ edges since it removes $c$ weight-$1$ edges, one from each cycle.
Hence, any $(1 + \epsilon)$-approximate MST has weight at most 
$$(1 + \epsilon)(n + c - 2) = n + \big(c + \epsilon(n + c - 2) \big) - 2,$$
and removes at most $c + \epsilon(n + c - 2)$ weight-$1$ edges.
\end{proof}

Given an MPC algorithm that finds a $(1 + \epsilon)$-approximation of the MST in $(1, 2)$-metric, we examine the probability of each weight-$1$ edge getting removed. This probability is used to determine the $\frac{1}{\epsilon}$-neighborhood of a vertex in \cref{clm:neighborhood-detection}. 
We keep referring to the subgraph of weight-$1$ edges as $G$.
We show that by a random reordering of the vertex IDs, any algorithm can be transformed such that for any two isomorphic edges, the probability of being removed by the algorithm is the same. 

\begin{definition}
    Given a graph $G$, two edges $(u_1, v_1)$ and $(u_2, v_2)$ are isomorphic if there is an automorphism of $G$ that maps one to the other.
    That is, if there is a bijection $\phi: V(G) \to V(G)$, such that $\phi(u_1) = u_2$, $\phi(v_1) = v_2$, and for any pair of vertices $u$ and $v$, there is an edge $(u, v)$ if and only if there an edge $(\phi(u), \phi(v))$.
\end{definition}

\begin{lemma} \label{lem:reordering}
    Given an MPC algorithm $\A$ that computes a $(1 + \epsilon)$-approximate MST with high probability, through a random reordeing of the vertex IDs, another algorithm $\A'$ can be obtained that solves the same problem using the same space and number of rounds,
    such that for any two isomorphic edges, the probability of being removed by $\A'$ is the same.
\end{lemma}
\begin{proof}
    The algorithm $\A'$ works as follows.
    It draws a permutation $\pi: V(G) \to V(G)$ uniformly at random.
    Then, it gives each vertex $u$ the  label $\pi(u)$, and runs $\A$ on the resulting graph, i.e.\ these labels are the new vertex IDs that $\A$ uses.
    
    Take any two isomorphic edges $e = (u, v)$ and $e' = (u', v')$, with automorphism $\phi: V(G) \to V(G)$, that is $u' = \phi(u)$ and $v' = \phi(v)$.
    Observe that any vertex $\phi(u)$ is labeled $\pi(u)$ if we relabel with $\pi \circ \phi^{-1}$.
    Therefore, since $e$ and $e'$ are isomorphic,
    the output of $\A'$ on $e$ when the permutation $\pi$ is drawn is the same as its output on $e'$ when permutation $\pi \circ \phi^{-1}$ is drawn,
    where by output we mean whether the edge is removed by the algorithm.
    Consequently, it holds that:
    \begin{align*}
        \Pr(\A' \textnormal{ removes } e)
        &= \sum_{\pi} \Pr(\pi \textnormal{ is drawn})\Pr(\A \textnormal{ removes } e \textnormal{ after relabling with } \pi) \\
        &= \sum_{\pi} \Pr(\pi \textnormal{ is drawn})\Pr(\A \textnormal{ removes } e' \textnormal{ after relabling with } \pi \circ \phi^{-1}) \\
        &= \sum_{\pi'} \Pr(\pi' \textnormal{ is drawn})\Pr(\A \textnormal{ removes } e' \textnormal{ after relabling with } \pi') \\
        &= \Pr(\A' \textnormal{ removes } e').
    \end{align*}
    Here, the first and last equalities hold by definition.
    The second holds because $e$ and $e'$ are isomorphic under $\phi$,
    and the third holds by the change of variable $\pi' = \pi \circ \phi^{-1}$. 
    This concludes the proof.
\end{proof}

With that, we are ready to prove \cref{clm:neighborhood-detection}.
\begin{proof}[Proof of \cref{clm:neighborhood-detection}]
Given $\frac{1}{\epsilon^k}$ cycles of length $\epsilon^k n$,
the goal is to find the $\frac{1}{\epsilon}$-neighborhood of each vertex.
We are given an algorithm $\A$ that computes a $(1 + \epsilon)$-approximate MST for $(1, 2)$-metrics with high probability, using $O(n^\delta)$ space per machine and $\poly(n)$ total space, in $R$ rounds.
Due to \cref{lem:reordering},
we can assume that $\A$ removes isomorphic edges with the same probability. Specifically, in a graph that is a disjoint union of cycles, the length of a cycle determines the removal probability of its edges.

As a warm-up, we show how to compute the $\frac{1}{\epsilon}$-neighborhoods assuming we can calculate the exact probability of each edge being removed. Later, we show how to remove this assumption by approximating the probabilities through sampling.

Given two vertices $u_1$ and $u_2$, we show how to determine whether they are within distance $D = \frac{1}{20\epsilon}$ of each other. We take a neighbor $v_1$ of $u_1$, and a neighbor $v_2$ of $u_2$, and obtain a new graph $G'$ by performing a two-switch on $(u_1, v_1)$ and $(u_2, v_2)$.
We do this for all the four possible ways of choosing the neighbors (each vertex has two possible neighbors),
and study the removal probabilities in each case when the algorithm $\A$ is run on $G'$.

The goal is to distinguish the case where $u_1$ and $u_2$ are at most $D$ edges apart, from the other cases. Roughly, if $u_1$ and $u_2$ are close, then doing the two-switch with the appropriate neighbors creates a small cycle. The edges of a small cycle have a larger removal probability than an average edge, and hence can be identified.
More formally, we examine the removal probabilities of all the edges after the two-switch, and use them to assert if $u_1$ and $u_2$ have distance at most $D$. A summary of the different scenarios appears in \Cref{tab:cases}.

Consider the case where $u_1$ and $u_2$ are in different cycles \textbf{(case 1)}. Doing the two-switch for any choice of neighbors will join the two cycles in $G'$.
That is, $G'$ contains $\frac{1}{\epsilon^k} - 2$ cycles of length $\epsilon^kn$, call their edges $E_1$, and one cycle of length $2\epsilon^k n$, call its edges $E_2$.
Observe that all the edges in $E_1$ are isomorphic, therefore the probability of each of them being removed is the same.
The same holds for $E_2$.
The removal probability in $E_1$ and $E_2$ may or may not be the same.
That is, either all the edges have the same removal probability \textbf{(case 1-a)}, or there are two groups with different probabilities: one of size $2\epsilon^k n$, and another of size $n - 2\epsilon^k n$ \textbf{(case 1-b)}.

Consider the case where $u_1$ and $u_2$ are in the same cycle, and $v_1$ and $v_2$ lie on different paths from $u_1$ to $u_2$ \textbf{(case 2)}.
In this case, after the two-switch, the graph still consists of $\frac{1}{\epsilon^k}$ cycles of length $\epsilon^k n$.
Therefore, since all the edges are isomorphic,
all the removal probabilities will be the same.

Consider the case where $u_1$ and $u_2$ are in the same cycle, and $v_1$ and $v_2$ lie on the same path from $u_1$ and $u_2$.
Doing the two-switch will break the cycle into two cycles, one containing $(u_1, u_2)$, and the other containing $(v_1, v_2)$.
First, we examine the situation where one of the two cycles has length at most $D$ \textbf{(case 3)}, say the one containing $(u_1, u_2)$ (the other scenario can be handled similarly).
Then, $G'$ consists of $\frac{1}{\epsilon^k} - 1$ cycles of length $\epsilon^k n$, call the edges $E_1$, one cycle of length at most $D$, call the edges $E_2$, and another cycle of length at least $\epsilon^k n - D$, call the edges $E_3$.
We analyze the removal probabilities in each of these edge sets.

The edges in $E_1$ are all isomorphic, and there are $n - \epsilon^k n \geq \frac{n}{2}$ of them. $\A$ computes a $(1 + \epsilon)$-approximation of the MST. Therefore, by \Cref{clm:edge-removal-apx-mst}, with high probability the number of edges it removes is at most (here, $c$ is the number of cycles in $G'$)
$$
c + \epsilon(n + c - 2)
= \frac{1}{\epsilon^k} + \epsilon\left( n + \frac{1}{\epsilon^k} - 2\right)
\leq \epsilon n + \epsilon(n + \epsilon n - 1)
\leq 3\epsilon n.
$$
Therefore, even if all these edges are removed from $E_1$, their removal probability is at most:
$$
\frac{3\epsilon n}{\card{E_1}} \leq \frac{3\epsilon n}{n / 2} \leq 6\epsilon.
$$
 
The edges of $E_2$ are isomorphic to each other, and there are at most $D = \frac{1}{20\epsilon}$ of them.
Since $\A$ computes a tree with high probability,
and $E_2$ is a cycle,
there must be at least one edge removed from $E_2$ with high probability.
Therefore, the removal probability in $E_2$ is at least:
$$
\frac{1/2}{D} \geq 10\epsilon.
$$

The edges of $E_3$ are all isomorphic. Their removal probability might be the same as $E_1$ \textbf{(case 3-a)}, the same as $E_2$ \textbf{(case 3-b)}, or different from both \textbf{(case 3-c)}.
As long as it is different from the removal probability in $E_2$, we can recognize that $u_1$ and $u_2$ are close seeing as that there is a group of at most $D$ edges that have a higher than $10 \epsilon$ removal probability.
If $E_3$ has the same removal probability as $E_2$, we need to investigate further.
In this case, the edges of $E_2$ and $E_3$ can be identified since they form a group of $\epsilon^k n$ edges with removal probability higher than $10 \epsilon$, and the rest of the edges have probability at most $6\epsilon$.
Putting together the vertices of $E_2$ and $E_3$, we can identify all the vertices that are in the same cycle of $G$ as $u_1$ and $u_2$.

Having identified the vertices of the cycle $C$ in $G$ that contains $u_1$ and $u_2$.
We obtain $G''$ by doing the same two-switch on $C$, i.e.\ on the edges $(u_1, v_1)$ and $(u_2, v_2)$. Note that the other vertices of $G$ are not included. Let $E'_2$ be the edges of the cycle in $G''$ that contains $u_1$ and $u_2$, and let $E'_1$ be the edges of the other cycle.
The key thing to notice is that this time since there are no other edges this time, $E'_2$ and $E'_3$ can be distinguished.

It can be shown using the same proof that the removal probability of $E'_2$ is still at least $10\epsilon$.
However, since $G''$ has $\epsilon^k n$ vertices and only two cycles,
the number of edges the algorithm removes from the $G''$ is at most (here, $c$ is the number of cycles in $G''$):
$$
c + \epsilon(\card{V(G'')} + c - 2) \leq 2 + \epsilon(\epsilon^k n + 2 - 2)
\leq 2 + \epsilon^{k+1}n.
$$
The number of edges in $E_2$ is at least $\epsilon^kn - \frac{1}{\epsilon} \geq
\frac{\epsilon^k n}{2}$, where the inequality holds because we are not in the base case and the length of the cycle is at least $\epsilon^kn \geq \frac{1}{\epsilon^2}$.
Therefore, the removal probability in $E'_2$ is at most:
$$
\frac{2 + \epsilon^{k+1}n}{\epsilon^k n / 2} \leq 3\epsilon.
$$
As a result, this case is also distinguishable, since there is one group of at most $D$ edges $E'_2$ with removal probability at least $10\epsilon$, and a group of at least $\epsilon^kn$ edges with removal probability at most $3\epsilon$.

The last case is where $u_1$ and $u_2$ are in the same cycle, $v_1$ and $v_2$ lie on the same path from $u_1$ to $u_2$, and both the cycles created after the two-switch have length larger than $D$ \textbf{(case 4)}. We define $E_1$, $E_2$, and $E_3$ similarly.
That is, $E_2$ is the edges of the cycle in $G'$ that contains $(u_1, u_2)$,
$E_3$ is the edges of the cycle in $G'$ that contains $(v_1, v_2)$,
and $E_1$ is the rest of the edges.
Using the same argument, one can show that the removal probability in $E_1$ is at most $6\epsilon$.
If $E_2$ and $E_3$ both have the same removal probability as $E_1$ \textbf{(case 4-a)}, then we can recognize $u_1$ and $u_2$ are further than $D$ apart, since all the edges have the same removal probability.
If $E_2$ and $E_3$ have the same removal probability as each other but different from $E_1$ \textbf{(case 4-b)}, then we can identify the vertices of the cycle in $G$ that contains $u_1$ and $u_2$, and treat it as in case 3-b.
Otherwise, if $E_2$ and $E_3$ have different removal probabilities \textbf{(case 4-c)}, we can recognize that $u_1$ and $u_2$ have distance more than $D$, since there is no group of fewer than $D$ edges that have the same removal probability.

\begin{table}
    \centering
    \resizebox{\textwidth}{!}{%
    \begin{tabular}{|c|c|c|l|}
\hline
 {\scriptsize  Same cycle} & {\scriptsize  Same path} & {\scriptsize  Length $\leq \frac{1}{20\epsilon}$} & {\scriptsize  Possible removal probabilities} \\
\hline
No & - & - & 
\begin{tabular}{l}
     All the edges have the same probability. \textbf{(1-a)} \\ Two distinguishable groups with sizes $2\epsilon^kn$ and $n - 2\epsilon^kn$. \textbf{(1-b)} 
\end{tabular} \\
\hline
 Yes & No & - &  
 \begin{tabular}{l}
 All the edges have the same probability.
 \textbf{(2)}
 \end{tabular}
 \\
\hline
 Yes & Yes & Yes &  \begin{tabular}{l}
    There are $n - \epsilon^k n$ edges with probability $\leq 6\epsilon$, \\
    $\leq \frac{1}{20\epsilon}$ edges with probability $\geq 10\epsilon$, \\
    and all the remaining edges have probabilities \\
    same as the first group \textbf{(3-a)}, \\
    same as the second group \textbf{(3-b)}, or \\
    different from both \textbf{(3-c)}.
 \end{tabular}\\
\hline
 Yes & Yes & No &  \begin{tabular}{l}
      All the edges have the same probabilities. \textbf{(4-a)}\\
      Two distinguishable groups of size $n - \epsilon^k n$ and $\epsilon^k n$. \textbf{(4-b)} \\
      Three distinguishable groups, one with $n - \epsilon^k n$ edges,\\ and the other two with $>\frac{1}{20\epsilon}$ edges. \textbf{(4-c)}
 \end{tabular}\\
\hline
\end{tabular}
}

    \caption{A summary of the possible scenarios. The columns denote: (1) whether $u_1$ and $u_2$ are in the same cycle in $G$;
    (2) if $v_1$ and $v_2$ lie on the same path from $u_1$ to $u_2$;
    (3) denotes whether one of the resulting cycles after the two-switch has length at most $\frac{1}{20\epsilon}$;
    and (4) the possible removal probabilities of the edges in that case.}
    \label{tab:cases}
\end{table}

Given two vertices $u_1$ and $u_2$, we can use these characterizations to detect if their distance is at most $D = \frac{1}{20\epsilon}$.
We do this for every pair of vertices, and then use $\ceil{\log 20}$ rounds of doubling to find the $\frac{1}{\epsilon}$-neighborhood of each vertex.

\paragraph{Approximating the removal probabilities:}
So far, we had assumed that given a graph $G'$,
we have access to the removal probability of each edge when $\A$ is run on the $(1, 2)$-metric derived from $G'$.
We lift that assumption here by approximating the probabilities through repeated sampling.

For each $G'$ obtained by doing a two-switch on $(u_1, v_1)$ and $(u_2, v_2)$,
we create $S = \Theta(\frac{\log n}{\epsilon^2})$ copies, and execute $\A$ independently on each copy in parallel.
For an edge $e \in G'$, let $p_e$ be the probability of $\A$ removing $e$, and let $\ph_e$ be the fraction of times $e$ was removed in the copies we executed. That is, $\ph_e$ is an approximation of $p_e$.
The additive version of the Chernoff bound gives:
\begin{equation}
\Pr(\card{\ph_e - p_e} > \epsilon/4) \leq 2e^{-32S\epsilon^2} \leq n^{-c}, \label{eq:prob-apx}
\end{equation}
where we can take $c$ to be an arbitrarily large constant.

Now, we group the edges based on their approximated removal probabilities. We call two edges $e_1$ and $e_2$ \emph{close} if $\card{\ph_{e_1} - \ph_{e_2}} < \frac{\epsilon}{2}$. 
We define the groups as the transitive closure of closeness.
That is, we let $e$ and $e'$ be in the same group if there is a chain of edges $e_0, e_1, \ldots, e_l$,
such that $e_0 = e$, $e_l = e'$, and $\card{\ph_{e_i} - \ph_{e_{i+1}}} \leq\frac{\epsilon}{2}$ for $0 \leq i \leq l-1$.
These groups can be computed in $O(1)$ rounds of MPC, by sorting the edges based on their approximated removal probabilities.

Observe that due to \eqref{eq:prob-apx}, if two edges have the same probability, then their approximated probability will differ by at most $\frac{\epsilon}{2}$ with high probability, in which case they will be grouped together.
Furthermore, if the probability of two edges differs by more than $\epsilon$, then they will not be close with high probability (however, they can still be grouped together through a chain of edges).
Consequently, since there are at most three distinct probabilities in all the cases,
if the probability of two edges differs by more than $2\epsilon$, then they will not be joined.
Specifically,
the edges with removal probability at least $10\epsilon$ will not be grouped with the edges that have removal probability at most $6\epsilon$, with high probability.

A pair of vertices $u_1$ and $u_2$ are within distance $\frac{1}{20\epsilon}$ of each other if and only if for some choice of neighbors, $v_1$ and $v_2$, case 3 happens.
Given the characterizations of the groups, and the different cases (\Cref{tab:cases}), we can tell when $(u_1, v_1)$ and $(u_2, v_2)$ are in case 3 (i.e.\ 3-a, 3-b, or 3-c) as follows.

\begin{enumerate}
    \item Given $(u_1, v_1)$ and $(u_2, v_2)$, obtain $G'$ by doing the two-switch on them.
    Then, obtain the approximated probabilities $\ph_e$ by running the $\A$ on $S = \Theta\left(\frac{\log n}{\epsilon^2}\right)$ instances of $G'$ independently. Group the edges as described above, and for each group, make note of the number of edges and the range of their approximated probabilities. \label{step:small-cycle}
    \item If there is a group of fewer than $\frac{1}{20\epsilon}$ edges, such that at least one of them has approximated probability larger than $9\epsilon$, report that this is case 3.\footnote{There is one corner case that has to be handled separately. If this group of edges with probability larger than $9\epsilon$ contains $(v_1, v_2)$, and it includes exactly $\frac{1}{20\epsilon}$ edges, then $u_1$ and $u_2$ have distance $\frac{1}{20\epsilon} + 1$ from each other.}
    \item Otherwise, if there is a group of exactly $\epsilon^k n$ edges such that at least one of them has approximated probability larger than $9\epsilon$, call their vertices $C$ (these vertices form a cycle in $G$).
    Obtain graph $G''$ by applying the same two-switch, i.e.\ on $(u_1, v_1)$ and $(u_2, v_2)$ to the graph $G[C]$.
    If $G''$ has a group of fewer than $\frac{1}{20\epsilon}$ edges, such that at least one of them has approximated probability larger than $9\epsilon$, report that this is case 3. \label{step:big-cycle}
    \item Otherwise, report that this is not case 3.
\end{enumerate}

This correctly distinguishes case $3$. For cases 3-a and 3-c, step \ref{step:small-cycle} reports that case 3 has happened.
Otherwise, step \ref{step:big-cycle}, further investigates the two cases 3-b and 4-b, and distinguishes between them.
All the other cases go to the last step, where it is reported case $3$ has not happened.
As explained, one can use this to tell whether $u_1$ and $u_2$ are within distance $\frac{1}{20\epsilon}$ of each other, which concludes the proof.
\end{proof}

\cref{clm:cycle-breakdown} follows easily from \cref{clm:neighborhood-detection}.
\begin{proof}[Proof of \cref{clm:cycle-breakdown}]
    We are given a graph $G$, comprising $\frac{1}{\epsilon^k}$ cycles $C_0, \ldots, C_{1/\epsilon^k-1}$ of length $\epsilon^k n$, and one vertex $u^i_0$ per cycle.
    Using \cref{clm:neighborhood-detection}, we compute for each vertex $u$, its $\frac{1}{\epsilon}$-neighborhood.
    This neighborhood can be collected in one machine, i.e.\ the subpath of length $\frac{2}{\epsilon}$ centered at the vertex is retrieved.

    To create $G'$, on the same vertex set, we connect each vertex to the two vertices at distance $\frac{1}{\epsilon}$ from it in $G$.
    Since the length of the cycles is divisible by $\frac{1}{\epsilon}$, each cycle $C_i$ is broken down into $\frac{1}{\epsilon}$ cycles $C^i_{0}, \ldots, C^i_{1/\epsilon - 1}$ of length $\epsilon^{k+1}n$.
    Also, $u^i_j$ is already computed for $0 \leq i \leq 1/\epsilon^k - 1$ and $0 \leq j \leq 1/\epsilon - 1$ using \cref{clm:neighborhood-detection}. It is the subpath starting at $u^i_0$.
\end{proof}

\section*{Acknowledgments}
Part of this work was conducted while the first and second named authors were visiting the Simons Institute for the Theory of Computing as participants in the Sublinear Algorithms program.

\printbibliography

\clearpage
\appendix
\section{Lower Bounds (LOCAL and Component-Stable MPC)}

In this section, we present lower bounds for MST approximation on $(1, 2)$-metrics.
First, we prove a lower bound for a slightly easier problem in the LOCAL model with shared randomness.
Then, we use the framework of \cite{GhaffariKU19} and \cite{CzumajDP21a} to lift it and get a $\Omega\left(\log \frac{1}{\epsilon}\right)$-round lower bound for component-stable MPC algorithms that compute a $(1 + \epsilon)$-approximate MST in $(1, 2)$-metrics.

We note that while the lower bound proved in this section is weaker than the lower bound of \cref{thm:lb} which does not require the component stability assumption, we state it nonetheless to contrast lower bounds obtained from this framework with our lower bound of \cref{thm:lb}.

\subsection{LOCAL Model}\label{subsec:lb-local-random}
\paragraph{The LOCAL Model:}
In the distributed LOCAL model, given a graph $G$,
there is a machine on each vertex.
Each machine has infinite memory and computation power, and it is given its adjacency list in the beginning.
The computation proceeds in synchronous rounds.
Each round has a processing phase, where each machine can do arbitrary computation on its data,
and a communication phase,
where each machine can send a message (of arbitrary length) to each of its neighbors.
In the end, each machine reports the part of the output that pertains to it.
In the LOCAL model with shared randomness, all machines can access an infinite random tape that is shared between them.

We restate some of the claims and definitions from \cref{sec:lb}.
Given a graph $G$, the $(1, 2)$-metric obtained from $G$ is such that for any two vertices $u$ and $v$, their distance in the metric is $1$ if and only if there is an edge between them in $G$.
We say a spanning tree $T$ of the metric removes an edge $e \in G$ if $e \notin T$ (\cref{def:removed-edges}). That is, $T$ is constructed by removing some edges from $G$, turning it into a forest, and then connecting the resulting components with weight-$2$ edges.
When $G$ is a disjoint union of cycles, a $(1 + \epsilon)$-approximate MST of this metric removes at least one edge from each cycle, and at most $c + \epsilon(n + c - 2)$ edges in total, where $n$ is the number of vertices, and $c$ is the number of cycles (\cref{clm:edge-removal-apx-mst}).

To model MST approximation on $(1, 2)$-metrics in the LOCAL setting, we relax the problem.
Given a graph $G$ that is a disjoint union of cycles, the approximate MST of the $(1, 2)$-metric obtained from $G$ is of interest.
The machines are allowed to communicate over $G$.
The LOCAL algorithm is required only to decide which edges of $G$ appear in the approximate MST and which are removed, i.e.\ it only needs to compute the weight-$1$ edges that appear in the approximate MST.
Formally, we define:
\begin{problem} \label{prob:almost-spanning-forest}
    In the $\epsilon$-\emph{almost-spanning forest} problem, we are given a parameter $0 < \epsilon < 1$, and a graph $G$ with $n$ vertices such that every component is either a cycle or a path. Let $c$ be the number of the cycles. The goal is to remove some of the edges such that
    \begin{enumerate}
        \item at least one edge in each cycle is removed, and
        \item at most $c + \epsilon n$ edges are removed.
    \end{enumerate}
    The set of remaining edges is called an $\epsilon$-almost-spanning forest.
\end{problem}
Observe that the second requirement is a relaxation of the original $c + \epsilon(n + c - 2)$ upper bound, up to a rescaling of $\epsilon$.
Also, the family of the graphs over which the problem is defined includes paths as components.
This is for the benefit of the next section. For now, one can think of the input graph as a disjoint union of cycles.

We prove the following lower bound:

\begin{theorem} \label{thm:lb-local-rand}
    Any LOCAL algorithm with shared randomness that solves \cref{prob:almost-spanning-forest} with high probability, takes at least $\Omega\left(\frac{1}{\epsilon}\right)$ rounds to finish.\footnote{The lower bound is still true even if the second requirement of \cref{prob:almost-spanning-forest} holds in expectation.}
\end{theorem}
\begin{proof}
    We make some definitions.
For a vertex $u$ in $G$, the $r$-neighborhood of $u$, denoted by $N_r(u)$, is the subgraph consisting of the vertices with distance at most $r$ from $u$ (identified by their vertex IDs), and the edges with at least one endpoint with distance at most $r-1$. Observe this is slightly different than the definition of $r$-neighborhood in \cref{sec:lb}. 

In the LOCAL setting, after $r$ rounds, at any vertex $u$, the knowledge of the graph is limited to $N_r(u)$.
Therefore, the output of any $r$-round LOCAL algorithm with shared randomness at a vertex $u$ can be expressed as a function of its $r$-neighborhood $N_r(u)$ and the shared random tape $\pi$, i.e.\ $f(N_r(u), \pi)$.\footnote{The proof still works if we allow a dependence on $n$, which is necessary for lifting the lower bound in the next section.}
The output of each node is whether each of its adjacent edges should be removed, and adjacent nodes must agree on whether the edge between them should be removed.

Let $G$ consist of two cycles: a cycle of length $\frac{1}{4\epsilon}$, and a cycle of length $n-\frac{1}{4\epsilon}$.
Intuitively, to solve the problem, each vertex must be able to recognize whether it is in a smaller cycle.
Because (approximately) a $4 \epsilon$ fraction of the edges in the smaller cycle must be removed, whereas only a $\epsilon$ fraction of the edges in the larger cycle can be removed. The claim is that $\Omega\left(\frac{1}{\epsilon}\right)$ rounds are needed to do so.

Take any algorithm that runs in $r = \frac{1}{10\epsilon}$ rounds.
Let $I$ be the set of vertex IDs in $G$,
and let $I_u$ be the ID of a vertex $u$.
Let $G'$ a graph obtained from $G$ by applying a uniformly random permutation $\pi: I \to I$  to the vertex IDs.
That is, vertex $u$ in $G'$ gets a new ID $\pi(I_u)$.
This creates a distribution $\D$ over the input graphs. We study the probability of the algorithm succeeding on a graph $G'$ drawn from $\D$.

For any two vertices $u$ and $v$, and edges $e_u$ and $e_v$ adjacent to $u$ and $v$ respectively,
the probability that $u$ removes $e_u$ in $G'$,
is the same as the probability that $v$ removes $e_v$.
Recall that the output at $u$ is a function of $N_r(u)$ and $\pi$, and the distribution of $(N_r(u), \pi)$ is the same as the distribution of $(N_r(v), \pi)$.
Specifically, for both vertices,
the neighborhood is a path of length $2r$ where each vertex has a random ID from $I$.

Take this probability to be $p$.
If $p \geq 2\epsilon$, then the algorithm does not satisfy the second requirement of \cref{prob:almost-spanning-forest}, that is does not remove at most $2 + \epsilon n$ edges with high probability.
If $p < 2\epsilon$, then the algorithm does not satisfy the first requirement,
since the smaller cycle has length $\frac{1}{4\epsilon}$ and removing the edges with probability $2\epsilon$ would mean that no edges are removed with probability $1 - 2\epsilon \cdot \frac{1}{4\epsilon} = \frac{1}{2}$. Observe, that if the algorithm solved the problem on any graph with high probability, it would have done so also on this distribution.
Therefore, any algorithm that runs in $\frac{1}{10\epsilon}$ rounds does not solve \cref{prob:almost-spanning-forest}
with high probability.
\end{proof}

\subsection{MPC Lower Bound}\label{sec:lift-to-MPC}
In this subsection, we prove a lower bound for \cref{prob:almost-spanning-forest} in the MPC model by using the framework of \citet*{GhaffariKU19} and \citet*{CzumajDP21a} to lift the lower bound in \cref{subsec:lb-local-random}. 

\subsubsection*{Background}
Here, we review the definitions and the main result of \cite{CzumajDP21a}. For our application, we have removed the dependence on the maximum degree from the statements.

Briefly explained, their framework studies labeling problems defined on normal graph families (closed under vertex deletion and disjoint union). Then, assuming the \onetwocycle{} conjecture, a randomized LOCAL lower bound of $\Omega(T(n))$ rounds implies a component-stable MPC lower bound of $\Omega(\log T(n))$ rounds, given that $T$ is a \enquote{well-behaved} (constrained) function.

\begin{definition}
    A \emph{vertex-labeling problem} is characterized by a possible set of labels $\Sigma$, and a set of valid labelings $\L_G$ for each graph $G$, where a labeling $L: V(G) \to \Sigma$ is a choice of label $L(u)$ for each vertex $u$. Given a graph $G$, the goal is to find a valid labeling $L \in \L_G$.
\end{definition}

They introduce the notion of component-stability which captures many natural MPC algorithms.
For a vertex $u$, let $CC(u)$ denote the topology of its component and the vertex-ID assignment on it.
Then, component stability is formally defined as follows:

\begin{definition}
    A randomized MPC algorithm is \emph{component-stable}
    if at any vertex $u$, the output is a function $n$ (the number of vertices), $CC(u)$ (the topology and the IDs of $u$'s component), and $\pi$ (the random tape shared between all machines).
\end{definition}

For the definition of component-stability to make sense,
for each vertex, a \emph{name} and an \emph{ID} is defined.
The names are unique, but the IDs are required to be unique only within each component.
The output of a component-stable algorithm may depend only on the IDs, while the names are used to distinguish between the vertices in the output.

The other definitions are rather technical, and the motivations behind them lie within the proof of \cref{thm:local-to-mpc} (the main theorem of \cite{CzumajDP21a}).

\begin{definition}
    Let $G$ be a graph such that $\card{V(G)} \geq 2$. Let $\Gamma_G$ be a graph with at least $\card{V(G)}^R$ copies of $G$ with identical vertex-ID assignment, and at most $\card{V(G)}$ isolated vertices with the same ID. Let $L$ be a labeling of $\Gamma_G$ such that all the copies of $G$ have the same labeling, and all the isolated vertices have the same labeling.
    Let $L'$ be the restriction of $L$ to a copy of $G$.
    
    A labeling problem $\P$ is \emph{$R$-replicable} if for any $G$, $\Gamma_G$, $L$, and $L'$ as defined above, whenever $L$ is a valid labeling for $\Gamma_G$,
    then $L'$ is a valid labeling for $G$. 
\end{definition}

\begin{definition}
    A family of graphs $\H$ is \emph{normal} if it is closed under vertex deletion and disjoint union.
\end{definition}

\begin{definition}
    A function $T : \mathbb{N} \to \mathbb{N}$ is constrained if $T = O(\log^\gamma(n))$ for some $\gamma \in (0, 1)$, and $T(n^c) \leq c \cdot T(n)$.\footnote{an alternative definition could be $T(n)^{-T(n)} = n^{-o(1)}$.}
\end{definition}

\begin{theorem} [Theorem 14 of \cite{CzumajDP21a}] \label{thm:local-to-mpc}
    Let $\P$ be an $O(1)$-replicable labeling problem that has $T(n)$ lower bound in the randomized LOCAL model with shared randomness, where $T$ is a constrained function, $n$ is the number of vertices, and $\P$ is defined on a normal family of graphs. Then, under the \onetwocycle{} conjecture,\footnote{equivalently, connectivity conjecutre}
    any component-stable MPC algorithm that solves $\P$ with success probability $1 - \frac{1}{n}$ must take $\Omega(\log T(n))$ rounds to finish.
\end{theorem}

\subsubsection*{Our Lower Bound}

We prove a lower bound for the $\epsilon$-almost-spanning forest problem in the (component-stable) MPC model. 

\begin{theorem} \label{thm:lb-mpc}
Let $\delta \in (0, 1)$ and $\epsilon = \Omega\left(\frac{1}{\poly \log n}\right)$.
Under the \onetwocycle{} conjecture, any component-stable MPC algorithm with $O(n^\delta)$ space per machine and total space $\poly(n)$ that solves the $\epsilon$-almost-spanning forest problem (\cref{prob:almost-spanning-forest}) with  probability $1 - \frac{1}{n}$ must take $\Omega\left(\log \frac{1}{\epsilon}\right)$ rounds to complete.
\end{theorem}

By a simple reduction, we can derive the following lower bound for MST approximation in $(1, 2)$-metrics:

\begin{theorem} \label{thm:lb-mst-mpc}
Let $\delta \in (0, 1)$ and $\epsilon = \Omega\left(\frac{1}{\poly \log n}\right)$.
Assuming the \onetwocycle{} conjecture, any MPC algorithm with $O(n^\delta)$ space per machine and total space $\poly(n)$ that computes a $(1 + \epsilon)$-approximation of MST in $(1, 2)$-metrics with  probability $1 - \frac{1}{n}$, and is component-stable on the weight-$1$ edges, must take $\Omega\left(\log \frac{1}{\epsilon}\right)$ rounds to complete.
\end{theorem}
\begin{proof}
    We give a reduction from the $\epsilon$-almost-spanning forest problem (\cref{prob:almost-spanning-forest}) to the $(1+\epsilon)$-approximation of MST in $(1, 2)$-metrics.
    Let $\A$ be an MPC algorithm as described in the statement that computes a $(1 + \epsilon/2)$-approximate MST. We derive an algorithm $\A'$ for \cref{prob:almost-spanning-forest}.
    Take an instance $G$ for \cref{prob:almost-spanning-forest}.
    Construct a $(1, 2)$-metric by letting all the pairs in $E(G)$ have weight $1$, and letting the rest of the pairs have weight $2$.
    Let $T$ be the output of $\A$ on the metric.
    We let the output of $\A'$ be $F$, the set of weight-$1$ edges in $T$, i.e.\ $F = T \cap G$.
    We claim that if $T$ is a $(1+\epsilon/2)$-approximate MST, then $F$ is an $\epsilon$-almost-spanning forest.
    
    First, observe that if $T$ is a $(1 + \epsilon/2)$-approximate MST, the first requirement of \cref{prob:almost-spanning-forest} is satisfied because $F$ is a subset of a tree, therefore, a forest.
    
    To see why the second requirement is satisfied, note that the weight of the minimum spanning tree in the metric is $2(n-1) - (\card{E(G)} - c)$, where $c$ is the number of cycles in $G$.
    Hence, $T$ has weight at most 
    \begin{align}
        w(T) 
        &\leq \left(1 + \frac{\epsilon}{2}\right)\Big(2(n-1) - (\card{E(G)} - c)\Big) \nonumber\\ 
        &= 2(n-1) - (\card{E(G)} - c) + \epsilon(n-1) - \frac{\epsilon}{2}\card{E(G)} + \frac{\epsilon}{2} c \nonumber\\
        &\leq 2(n-1) - (\card{E(G)} - c) + \epsilon(n-1) \label{eq:c-leq-E}\\
        &= 2(n-1) - \Big(\card{E(G)} - (c + \epsilon(n-1))\Big).\nonumber
    \end{align}
    Here, \eqref{eq:c-leq-E} follows from $c \leq \card{E(G)}$.
    This implies that $F$, in the notion of \cref{prob:almost-spanning-forest}, removes at most $c + \epsilon(n-1)$ weight-$1$ edges, i.e.\ $\card{G \setminus F} \leq c + \epsilon(n-1)$, and the second requirement is satisfied.

    In conclusion, this reduction provides an algorithm $\A'$ for \cref{prob:almost-spanning-forest}. Because whenever algorithm $\A$ outputs a $(1 + \epsilon/2)$-approximate $T$,
    then $\A'$ outputs an $\epsilon$-almost-spanning forest. Also, note that since $\A$ is component-stable on the weight-$1$ edges, so is $\A'$. This completes the reduction, and the lower bound is implied from \cref{thm:lb-mpc}.
\end{proof}

The remainder of this subsection is devoted to the proof of \cref{thm:lb-mpc}. We show that \cref{thm:local-to-mpc} applies to \cref{prob:almost-spanning-forest} and \cref{thm:lb-local-rand} through a series of claims.

\begin{claim} \label{clm:is-labeling}
    \cref{prob:almost-spanning-forest} can be expressed as a labeling problem.
\end{claim}
\begin{proof}
    We let the set of possible labels be $\{0, 1\}^2$, indicating whether each edge of the vertex is removed or not.
    Then, a labeling is valid if:
    \begin{enumerate}
        \item when an edge is absent, i.e.\ the vertex has only one or zero edges, then the label indicates the edge as removed,
        \item adjacent vertices agree on whether the edge between them is removed, and
        \item the two requirements of \cref{prob:almost-spanning-forest} are satisfied. \qedhere
    \end{enumerate}
\end{proof}

\begin{claim} \label{clm:is-normal}
    \cref{prob:almost-spanning-forest} is defined on a normal family of graphs.
\end{claim}
\begin{proof}
    The family of graphs $\H$ on which the problem is defined, consists of graphs such that all the components are either cycles or paths.
    To see why $\H$ is normal,
    observe that after taking the disjoint union of two of these graphs, all the components are still cycles or paths.
    Therefore, $\H$ is closed under disjoint union.
    Also, deleting a vertex from a cycle turns it into a path,
    and deleting a vertex from a path
    turns it into at most two paths. Consequently, $\H$ is also closed under vertex deletion, which concludes the proof.
\end{proof}

\begin{claim} \label{clm:is-replicable}
    Finding an $O(\epsilon)$-almost-spanning forest is $1$-replicable.
\end{claim}
\begin{proof}
    Let $G \in \H$ be a graph with $n$ vertices an $c$ cycles.
    Let $\Gamma_G$ be a graph with $k \geq n$ copies of $G$ and at most $n$ extra isolated vertices.
    Let $L$ be valid labeling of $\Gamma_G$ such that every copy of $G$ has the same labeling, and let $L'$ be the restriction of $L$ to a copy of $G$. To prove the claim, we need to show that $L'$ is valid for $G$. To do so, we check the two requirements of \cref{prob:almost-spanning-forest}.

    The first requirement is satisfied because otherwise there is a cycle in $G$ where $L'$ does not remove any edges. Therefore, $L$ does not remove any edges from the copies of the cycle either, and $L$ is invalid, which is a contradiction.

    To see why the second requirement is satisfied, observe that $L$ removes at most $kc + (k+1)\epsilon n$ edges from $\Gamma_G$.
    Therefore, the number of edges $L'$ removes from $G$, which is the number of edges $L$ removes from one copy of $G$, is at most:
    $$
    \frac{ck + (k+1)\epsilon n}{k} = c + \left(1 + \frac{1}{k}\right)\epsilon n \leq c + 2 \epsilon n.
    $$
    Hence, the second requirement holds but with a parameter that is a constant factor larger, namely $2\epsilon$. This is sufficient for the purposes of \cref{thm:local-to-mpc} (see the proof of Lemma 25 in \cite{CzumajDP21a} where this property is utilized, or Lemma 11 in the same paper for another example where the parameter is larger by a constant factor), and the problem is $1$-replicable.
\end{proof}

Putting these together gives \cref{thm:lb-mpc}.

\begin{proof}[Proof of \cref{thm:lb-mpc}]
First, we prove the theorem for $\epsilon = \Omega\left(\frac{1}{\log^\gamma n}\right)$, where $\gamma \in (0, 1)$. Note that $\frac{1}{\epsilon}$ is constrained. 
Due to \cref{thm:lb-local-rand}, we have a $\Omega\left(\frac{1}{\epsilon}\right)$ round lower bound for \cref{prob:almost-spanning-forest} in the randomized LOCAL model with shared randomness. Therefore, by \cref{clm:is-labeling,clm:is-normal,clm:is-replicable}, we can apply \cref{thm:local-to-mpc} and get a $\Omega\left(\log \frac{1}{\epsilon}\right)$ round lower bound for component-stable MPC algorithms, which concludes the proof.

For the cases with $\epsilon = \Theta\left(\frac{1}{\log^a n}\right)$, where $a \geq 1$, observe that any $\left(\frac{1}{\log^a n}\right)$-almost-spanning forest is a $\left(\frac{1}{\sqrt{\log n}}\right)$-almost-spanning forest.
Therefore, we get a $\Omega\left(\log \sqrt{\log n}\right)$, which is the same as $\Omega\left(\log \log^a n\right)$ and concludes the proof.
\end{proof}

\section{Implications for TSP}
\label{sec:tsp}

In the \emph{traveling salesman problem (TSP)}, given a weighted graph $G$, the goal is to find the shortest cycle that includes all the vertices (exactly once). When combined with the techniques of \cite{jayaram2024massively}, \cref{thm:main} leads to the following corollary for TSP.

\begin{theorem}
Given a metric, for any fixed $\delta \in (0, 1)$ and any $\epsilon > 0$, a $(2 + \epsilon)$-approximate TSP can be computed in $O\left(\log \frac{1}{\epsilon} + \log \log n\right)$ rounds of the MPC model, with $O(n^\delta)$ space per machine and $\widetilde{\Theta}(n^2)$ total space.
\end{theorem}

\begin{proof}
In metrics, any closed walk can be converted to a closed cycle with the same vertex set and at most the same weight.
To do so, it suffices to go over the vertices of the walk one by one and append each vertex to the output sequence if it has not already been added.
As the distances between the vertices satisfy the triangle inequality,
skipping over the repeated vertices does not increase the cost.
Therefore, the weight of the resulting cycle is not larger than the weight of the original walk.
This can be implemented in the MPC model in $O(1)$ rounds.

Since $\MST \leq \TSP$, any Eulerian tour of a $(1+\epsilon)$-approximate MST provides a $(2 + 2\epsilon)$-approximate TSP,
where an Eulerian tour of a tree $T$ is a closed walk on the edges of $T$ that includes each edge exactly twice.
\citet*{jayaram2024massively} showed that the Eulerian tour of a tree can be computed in few rounds of the MPC model provided that the tree satisfies certain hierarchical properties:

\begin{lemma}[Theorem 16 of \cite{jayaram2024massively}]
\label{lem:jmnz-tsp}
    Given a set of vertices $V,$ let $\P_0, \P_1, \ldots, \P_L$ be a hierarchy of partitions such that $\P_0 = \{ \{u\} \mid u \in V\}$ and $\P_L = \{V\}$.
    Let $T = \bigcup_{1 \leq i \leq L} E_i$ be a tree such that $\P_i = \P_{i-1} \oplus E_i$ and $\card{\P_{i-1}} - \card{\P_i} = \card{E_i}$.
    For each level $i$, consider the forest $F_i$ obtained by taking $E_i$ and contracting all the vertex sets of $\P_{i - 1}$, and let $\Lambda$ be an upperbound for the unweighted diameter of the components in $F_i$. 
    There exists an MPC algorithm that computes an Eulerian tour of $T$ in $O(\log L + \log \Lambda)$ rounds, using total space $O(nL + n^{1 + \delta})$ and space $O(n^\delta)$ per machine, where $\delta \in (0, 1)$ is an arbitrary constant.
\end{lemma}

The hierarchy $\{\Ph_t\}_{t = \alpha^k}$ and the edge sets $\{E_t\}_{t = \alpha^k}$, computed by \Cref{alg:main}, satisfy the hierarchical requirements.
The number of levels is $L = \log_{\alpha}n = O(\log n)$.
To bound $\Lambda$, observe that for each level $t = \alpha^k$, $E_t$ is obtained by performing $O\left(\log \frac{1}{\epsilon} + \log \log n\right)$ rounds of the modified \Boruvka{} algorithm where each round consists of computing and contracting a set of stars (step \ref{step:mst-boruvka}), and the edges computed in step \ref{step:mst-join-arbitrarily} can be taken to be a star.
Therefore, $\Lambda = \exp\left( O\left(\log \frac{1}{\epsilon} + \log \log n\right) \right)$.

To conclude the proof,
we can first compute the $(1 + \epsilon)$-approximate MST along with a hierarchy of partitions by utilizing \Cref{alg:main}, in $O\left(\log \frac{1}{\epsilon} + \log \log n\right)$ rounds, using total space $\Ot(n^2)$.
Then, by applying \cref{lem:jmnz-tsp}, we can obtain the $(2 + \epsilon)$-approximate TSP, in $O\left(\log \frac{1}{\epsilon} + \log \log n\right)$ rounds, using total space $O(n^{1 + \delta})$.
That is $O\left(\log \frac{1}{\epsilon} + \log \log n\right)$ rounds and $\Ot(n^2)$ total space overall, as desired. 
\end{proof}

\end{document}